\documentclass[11pt,a4paper,reqno]{amsart}%
\usepackage{amsthm,amsmath,amsfonts,amssymb,xcolor,amsxtra,appendix,bookmark,dsfont,bm,mathrsfs}
\usepackage[latin1]{inputenc}
\usepackage{color} 
\usepackage{amssymb}\usepackage{graphicx}
\usepackage{tikz}
\usepackage{hyperref}
\theoremstyle{plain}
\newtheorem{theorem}{Theorem}
\newtheorem{lemma}[theorem]{Lemma}

\newtheorem{proposition}[theorem]{Proposition}

\newtheorem{definition}{Definition}

\theoremstyle{remark}
\newtheorem{remark}[theorem]{Remark}

\allowdisplaybreaks

\title[On the energy transfer]{On the energy transfer towards large values of wavenumbers for solutions of 4-wave kinetic equations}

\author[G. Staffilani]{Gigliola Staffilani
}
\address{Department of Mathematics, Massachusetts Institute of Technology, Cambridge, MA 02139, USA}
\email{gigliola@math.mit.edu} 
\thanks{G.S. is  funded in part by  the NSF grants DMS-2052651, DMS-2306378 and the Simons Foundation through the Simons Collaboration on Wave Turbulence.}

\author[M.-B. Tran]{Minh-Binh Tran}
\address{Department of Mathematics, Texas A\&M University, College Station, TX 77843, USA}
\email{minhbinh@tamu.edu} 
\thanks{M.-B. T is  funded in part by  a   Humboldt Fellowship,   NSF CAREER  DMS-2044626/DMS-2303146, and NSF Grants DMS-2204795, DMS-2305523,  DMS-2306379.}

\begin{document}
\date{\today}

\begin{abstract} 
Inspired by the fundamental work of Escobedo and Velazquez \cite{EscobedoVelazquez:2015:FTB,EscobedoVelazquez:2015:OTT}, we prove that solutions of 4-wave kinetic equations, under very general forms of the dispersion relations, exhibit the transfer of energy towards large values of wavenumbers as time evolves. We also provide a global well-posedness theory for mild solutions of the equation under   general forms of the dispersion relations.
\end{abstract}

\maketitle

 \tableofcontents
 
\section{Introduction}\label{intro} 
Wave turbulence  describes the dynamics   of spectral energy transfer through  probability densities in weakly non-linear   wave systems. Those probability densities are solutions of  several types of wave kinetic equations.
The foundation of this theory started  with   works of  Peierls \cite{Peierls:1993:BRK,Peierls:1960:QTS}, Brout-Prigogine \cite{brout1956statistical},  Zaslavskii-Sagdeev \cite{zaslavskii1967limits}, Hasselmann \cite{hasselmann1962non,hasselmann1974spectral},  Benney-Saffman-Newell \cite{benney1969random,benney1966nonlinear},  Zakharov \cite{zakharov2012kolmogorov}. We refer to the books   \cite{zakharov2012kolmogorov,Nazarenko:2011:WT} for more discussions on the topics.

We consider the 4-wave kinetic equation 
\begin{equation}\label{4wave}
	\begin{aligned}
		\partial_t f\ = \ & 	\mathcal Q \left[ f\right],\ \ \ f(0,k)=f_0(k),\ \ \ k\in\mathbb{R}^3,\ \ \ t\in [0,\infty),\\ 
		\mathcal Q \left[ f\right]  
		\ =\ & \iiint_{\mathbb{R}^{3\times3}}\mathrm{d}k_1\,\mathrm{d}k_2\,\mathrm{d}k_3  \delta(k+k_1-k_2-k_3)\delta(\omega + \omega_1 -\omega_2 - \omega_3)[f_2f_3(f_1+f)-ff_1(f_2+f_3)] ,
	\end{aligned}
\end{equation}
where $f,f_1,f_2,f_3$ stand for $f(k), f(k_1), f(k_2), f(k_3)$, $\omega,\omega_1,\omega_2,\omega_3$ stand for $\omega(k), \omega(k_1), \omega(k_2), \omega(k_3)$.

In the pioneering and ground breaking works of Escobedo-Velazquez \cite{EscobedoVelazquez:2015:FTB,EscobedoVelazquez:2015:OTT}, numerous  fundamental results have been obtained for the case $\omega(k)=|k|^2$.

The current manuscript is the first  in our series of 2 papers that are devoted to the  study of equation \eqref{4wave}, when  the dispersion relation $\omega$ is allowed to take a very general form.  Since the dispersion relation $\omega$ is of very general form instead of the classical one  $\omega(k)=|k|^2$, and hence  significantly harder to treat. In this paper, our first result  is to provide a global well-posedness theory for mild solutions of the equation. Then, we prove that the energy of the solutions go to high wavenumbers as time evolves. To be more precise, we prove that

	\begin{equation}\label{Intro:1}\begin{aligned}
		&\lim_{t\to\infty}\int_{\mathbb{R}^3}	 \mathrm dk f\left(t,k\right)   \omega(k) \chi_{B(O,\mathfrak R)}(k)
		\ = \ 0,\end{aligned}
\end{equation}
for any $\mathfrak R>0$. The proof of \eqref{Intro:1} is based on several new ideas. First, we need to study the collisional region of the collision operator $\mathcal Q$ as time evolves. The study of the geometries of collisional regions  is done via a packing and covering technique.  Next, when $\omega(k)=|k|^2$, the conservation laws
$ k+k_1=k_2+k_3,   |k|^2+|k_1|^2=|k_2|^2+|k_3|^2$
play a very important role. This essentially means that $k$, $k_1,$ $k_2,$ $k_3$ are on the sphere centered at $\frac{k+k_1}{2}$ with radius $\frac{|k-k_1|}{2}$. The  collision operator can be interpreted  as integrals on spheres.  When $\omega$ is allowed to take the general form, the  forms of the dispersion relations make the collision operators become integrals on much more complicated resonance manifolds. We then develop further the parametrization  technique introduced in \cite{ToanBinh}  to tackle the problem. Moreover, for general types of dispersion relations, quantifying the collision region, which is the domain where the collisions happen, is a difficult problem. In this work, we follow the strategies used in the packing and covering problem (see, for instance \cite{bourgain1991convering,erdHos1964amount,fejes1985stable,toth2014regular,rogers1957note,rogers1963covering,rogers1964book}) to study this technical issue. The transfer of energy can be proved via the combination of the parametrization technique, a packing and covering lemma  and the choices of different test functions. 

In the second paper of our series  \cite{StaffilaniTranCascade2}, we will study the question of finite time condensations i.e. the formation of a delta function at the origin in finite time.

We now briefly summarize what is known about the wave kinetic equation \eqref{4wave} and related ones. The local wellposedness theory of \eqref{4wave} has been studied in \cite{germain2017optimal}. The stability and cascades for the Kolmogorov-Zakharov spectrum  and the near equilibrium stability, instability have been  investigated in  \cite{collot2024stability}  and \cite{menegaki20222,escobedo2024instability}. The convergence rates of discrete solutions and the local well-posedness for the MMT model      have been obtained in \cite{dolce2024convergence}  and \cite{germain2023local}. 
Besides the 4-wave kinetic equations,    the 3-wave kinetic equation has also  been extensively studied. We mention     the works \cite{AlonsoGambaBinh,CraciunBinh,EscobedoBinh,GambaSmithBinh,tran2020reaction} on the phonon interactions in anharmonic crystal lattices, the works \cite{nguyen2017quantum,soffer2020energy,walton2022deep,walton2023numerical,walton2024numerical} on capillary waves, the work     \cite{GambaSmithBinh} on stratified  flows in the ocean, the works \cite{cortes2020system,EPV,escobedo2023linearized1,escobedo2023linearized,nguyen2017quantum,soffer2018dynamics,PomeauBinh} on Bose-Einstein Condensates, and the work \cite{rumpf2021wave} on beam waves. Important progress has been obtained in rigorously justifying the derivation of wave kinetic equations     \cite{buckmaster2019onthe,buckmaster2019onset,collot2020derivation,collot2019derivation,deng2019derivation,deng2021propagation,deng2023long,deng2022wave,dymov2019formal,dymov2019formal2,dymov2020zakharov,dymov2021large,germain2024universality,grande2024rigorous,hani2023inhomogeneous,hannani2022wave,LukkarinenSpohn:WNS:2011,staffilani2021wave,ma2022almost}.

{\bf Acknowledgment} We would like to express our gratitude to Prof. J. J.-L. Velazquez, Prof. H. Spohn, Prof. J. Lukkarinen, Prof. M. Escobedo for fruitful discussions on the topics.

\section{The Settings}
 For a vector $k\in\mathbb{R}^3$, suppose that $\omega(k)=\omega(|k|)$ is convex in $|k|$.  We define
\begin{equation}
	\label{Settings1}\mho=\frac{|k|}{\omega'(|k|)}.
\end{equation}  
We suppose that $\mho$ is a function of the variables $\omega$ and $|k|$ from $[0,\infty)$ to $[0,\infty)$. In addition, 
  $\mho$ is  increasing in both $|k|$ and $\omega$
	and  
	\begin{equation}
		\label{Settings2} \mho\le {C}_\mho^1 |k|^\iota, \mbox{ for all } k, \end{equation}   where ${C}_\mho^1\ge0; 1\ge\iota\ge 0$ are constants independent of $k$; and
	\begin{equation}
		\label{Settings3}\begin{aligned}
		&\omega(|k|) \ge C_\omega|k|^\alpha, \forall k\in\mathbb{R}^d,\\
		&\omega(|k|) \le C_\omega'|k|^{\alpha'}, \forall k\in\mathbb{R}^d,|k|<1, 
		\end{aligned}
	\end{equation}  
	for some constants $2\ge\alpha >1,\alpha\ge \alpha'> 1,C_\omega,C_\omega'>0$.
 
\begin{remark}
Below are some physical examples of $\omega$ that satisfie the above assumptions: $\omega(k)=|k|^\alpha$ for $1<\alpha\le 2$. We have $\mho=\frac1\alpha|k|^{2-\alpha}$.\end{remark}
\begin{remark} The case when $\omega(k)=|k|^\alpha$ with $2<\alpha$ is  not covered by our theorem since the solutions exhibit a different behavior. We will investigate this behavior in a different work.
\end{remark}

\begin{definition}
	We say that $f(t,k)=f(t,|k|)\in C^1([0,\infty), L^1(\mathbb{R}^3))$ is a global mild radial solution of \eqref{4wave} with a radial initial condition $f_0(k)=f_0(|k|) \ge0$ if $f(t,k)\ge0$ and
	for all $\phi\in C_c^2([0,\infty))$, we have
	\begin{equation}\label{4wavemild}
		\begin{aligned}
			\int_{\mathbb{R}^3}\mathrm{d}k f(t,k)\phi(|k|)\ = \ & 	\int_{\mathbb{R}^3}\mathrm{d}k f(0,k)\phi(|k|) \  + \ \int_0^t\int_{\mathbb{R}^3}\mathrm{d}k	\mathcal Q \left[ f\right]\phi(|k|),
		\end{aligned}
	\end{equation}
	for all $t\in\mathbb{R}_+$. 
\end{definition}

Our main theorem is stated as  follows.
\begin{theorem}
	\label{Theorem:Main} Suppose that $\omega$ satisfies the above assumptions. 
	 Let $f_0(k)=f_0(|k|)\ge 0$ be an initial condition satisfying 
	\begin{equation}
		\label{Theorem:Main:1} \int_{\mathbb{R}^3}\mathrm{d}k f_0(k) \ = \ \mathscr{M}, \ \ \ \ \ \int_{\mathbb{R}^3}\mathrm{d}k f_0(k) \omega(k) \ = \ \mathscr{E}. 
	\end{equation}
	There exists at least a global mild radial solution $f(t,k)$ of \eqref{4wave} in the sense of \eqref{4wavemild} such that
	\begin{equation}
		\label{Theorem:Main:2} \int_{\mathbb{R}^3}\mathrm{d}k f(t,k) \ = \ \mathscr{M},   \ \ \ \ \ \int_{\mathbb{R}^3}\mathrm{d}k f(t,k) \omega(k) \ = \ \mathscr{E},
	\end{equation}
	for all $t\ge 0$. Suppose further that 
	\begin{itemize}
		\item [(I)] When $\omega(k)=|k|^2$, the support of $f_0$ has a non-empty interior;
		\item [(II)] When $\omega(k)\ne|k|^2$,   the origin belongs to the interior of the support of $f_0$.
	\end{itemize}
	For any $\mathfrak R>0$, we have the following energy cascade phenomenon  
	\begin{equation}\label{Theorem:Main:3}\begin{aligned}
			&\lim_{t\to\infty}\int_{\mathbb{R}^3}	  d\omega f\left(t,k\right)   \omega(k) \chi_{B(O,\mathfrak R)}(k)
			\ = \ 0.\end{aligned}
	\end{equation}

\end{theorem}

\subsection{A new formulation for the collision operator}
In this section, we introduce a new formulation for the collision operator.
We first compute

\begin{equation*}
	\begin{aligned}
		\mathcal Q \left[ f\right]  
		\ =\ & \iiint_{\mathbb{R}^{3\times3}}\mathrm{d}k_1\,\mathrm{d}k_2\,\mathrm{d}k_3  \delta(k+k_1-k_2-k_3)\delta(\omega + \omega_1 -\omega_2 - \omega_3)[f_2f_3(f_1+f)-ff_1(f_2+f_3)] \\
		\ =\ & \iiint_{\mathbb{R}_+^{3}}\mathrm{d}|k_1|\,\mathrm{d}|k_2|\,\mathrm{d}|k_3|  \delta(\omega + \omega_1 -\omega_2 - \omega_3)[f_2f_3(f_1+f)-ff_1(f_2+f_3)]|k_1|^2|k_2|^2|k_3|^2\\				
		&\times \iiint_{\left(\mathbb{S}^2\right)^3}\mathrm{d}\mathcal O_1\mathrm{d}\mathcal O_2\mathrm{d}\mathcal O_3\left[\frac1{(2\pi)^3}\int_{\mathbb{R}^3}\mathrm{d}\tau e^{i\tau\cdot({k}+{k}_1-{k}_2-{k}_3)}\right]\\
		\end{aligned}
\end{equation*}
\begin{equation*}
	\begin{aligned}
		\ =\ & \iiint_{\mathbb{R}_+^{3}}\mathrm{d}|k_1|\,\mathrm{d}|k_2|\,\mathrm{d}|k_3| \delta(\omega + \omega_1 -\omega_2 - \omega_3)[f_2f_3(f_1+f)-ff_1(f_2+f_3)]|k_1||k_2||k_3|\\	
		&\times \frac{32\pi}{|k|}\int_0^\infty\mathrm{d}|\tau|\sin(|k_1||\tau|)\sin(|k_2||\tau|)\sin(|k_3||\tau|)\sin(|k||\tau|)\frac{1}{|\tau|^2}\\
		\end{aligned}
\end{equation*}	
\begin{equation}\label{Lemma:TestFunction:E1}
	\begin{aligned}	
		\ =\ & \iiint_{\mathbb{R}_+^{3}}\mathrm{d}|k_1|\,\mathrm{d}|k_2|\,\mathrm{d}|k_3| \delta(\omega + \omega_1 -\omega_2 - \omega_3)[f_2f_3(f_1+f)-ff_1(f_2+f_3)]|k_1||k_2||k_3|\\
		&\times \frac{32\pi}{|k|}\int_0^\infty\mathrm{d}x\sin(|k_1|x)\sin(|k_2|x)\sin(|k_3|x)\sin(|k|x)\frac{1}{x^2}.
	\end{aligned}
\end{equation}
Next, we compute the integral in $x$ in \eqref{Lemma:TestFunction:E1}. To this end, we  write  
\begin{equation}\label{Lemma:TestFunction:E2}
	\begin{aligned}
		& \int_0^\infty\mathrm{d}x\sin(|k_1|x)\sin(|k_2|x)\sin(|k_3|x)\sin(|k|x)\frac{1}{x^2}\\
		=\	& \int_0^\infty\mathrm{d}x\frac1{16}\prod_{j=1}^4(e^{i|k_j| x}-e^{-i|k_j| x})\frac{1}{x^2}\\
		\end{aligned}
\end{equation}
\begin{equation*}
	\begin{aligned}
		=\	& \int_0^\infty\mathrm{d}x\frac{1}{x^2}\sum_{s_1,s_2,s_3,s_4=1}^2\frac1{16}\prod_{j=1}^4(-1)^{s_j}e^{i(-1)^{s_j}|k_j|x}\\
		=\	& \int_0^\infty\mathrm{d}x\frac{1}{x^2}\sum_{s_1,s_2,s_3,s_4=1}^2\frac1{16}(-1)^{s_1+s_2+s_3+s_4} \exp\Big\{i\Big(\textstyle\sum_{j=1}^4(-1)^{s_j}|k_j|\Big)x\Big\}\\
		\end{aligned}
\end{equation*}
\begin{equation*}
	\begin{aligned}
		=\	&		\frac\pi{16}\sum_{s_1,s_2,s_3,s_4=1}^2(-1)^{s_1+s_2+s_3+s_4+1}\Big|\sum_{j=1}^4(-1)^{s_j}|k_j|\Big|\\
		=\ &\frac\pi{8}\Big[-||k_1|+|k_2|+|k_3|+|k_4||+||k_1|+|k_2|+|k_3|-|k_4||+||k_1|+|k_2|-|k_3|+|k_4||\\
		&+||k_1|-|k_2|+|k_3|+|k_4||+ 	||k_1|-|k_2|-|k_3|-|k_4||\\
		& -||k_1|-|k_2|+|k_3|-|k_4||-||k_1|+|k_2|-|k_3|-|k_4||-||k_1|-|k_2|-|k_3|+|k_4||\Big]\\
		=\ &  \tfrac\pi4\min\{|k_1|,|k_2|,|k_3|,|k|\}.
	\end{aligned}
\end{equation*}

Combining \eqref{Lemma:TestFunction:E1} and \eqref{Lemma:TestFunction:E2}, we find
\begin{equation}\label{Lemma:TestFunction:E3}
	\begin{aligned}
		\mathcal Q \left[ f\right]  
		\ =\ & 8\pi^2\iiint_{\mathbb{R}_+^{4}}\mathrm{d}|k_1|\,\mathrm{d}|k_2|\,\mathrm{d}|k_3| \delta(\omega + \omega_1 -\omega_2 - \omega_3)[f_2f_3(f_1+f)-ff_1(f_2+f_3)]\\
		&\times \frac{|k_1||k_2||k_3|\min\{|k_1|,|k_2|,|k_3|,|k|\}}{|k|}.
	\end{aligned}
\end{equation}
As we do not need to keep tract of the constant  $8\pi^2$, we replace it by a universal constant $C_Q$
\begin{equation}\label{Lemma:TestFunction:E4}
	\begin{aligned}
		\mathcal Q \left[ f\right]  
		\ =\ & C_Q\iiint_{\mathbb{R}_+^{3}}\mathrm{d}|k_1|\,\mathrm{d}|k_2|\,\mathrm{d}|k_3|   \delta(\omega + \omega_1 -\omega_2 - \omega_3)[f_2f_3(f_1+f)-ff_1(f_2+f_3)]\\
		&\times \frac{|k||k_1||k_2||k_3|\min\{|k_1|,|k_2|,|k_3|,|k|\}}{|k|^2}.
	\end{aligned}
\end{equation}
\section{A packing and covering lemma and the support of the solutions }
\subsection{A packing and covering lemma}
First, we define the so-called collisional region. Note that the geometries of collisional regions  for  3-wave kinetic collision operators were also studied in \cite{rumpf2021wave}.  
\begin{definition}
	\label{Def:CollisionRegion}
	Let $\mathcal{S}$ be a set in $\mathbb{R}^3$. We denote the origin by $O$. We define the set $\mathcal{S}_n$ inductively as follow:
	  $\mathcal{S}_0[S]:=\mathcal{S}\backslash\{O\}$, $\mathcal{S}_{n+1}[\mathcal{S}]:=\{k\in \mathbb{R}^d ~~|~~ k=k_1+k_2-k_3; \omega_k=\omega_{k_1}+\omega_{k_2}-\omega_{k_3};k_1,k_2,k_3\in\mathcal{S}_n[\mathcal S];k_1+k_2\ne 0\}\backslash\{O\}$. The collisional region is defined as 
	  $$\mathrm{Coll}[\mathcal{S}] \ := \ \bigcup_{n=0}^\infty \mathcal{S}_n[\mathcal S].$$
\end{definition} 

\begin{definition} 	\label{Def:Interior} Below, we recall some standard definitions that we will use in the paper 
\begin{itemize}
	\item[(A)] 	Let $\mathcal{S}$ be a set in $\mathbb{R}^3$. The interior $\dot{A}$ of $A$ is defined as
	$$\dot{A} \ := \ \Big\{x\in A ~~| ~~ B(x,r)\subset A \mbox{ for all sufficiently small } r>0\}.$$
	Moreover, $\partial A = \bar{A}\backslash \dot{A}$. 
	\item[(B)] For a given point $x_0\in\mathbb{R}^3$. Let $\sigma$ be a unit vector of $\mathbb{R}^3$ and we suppose that $\sigma$ has one end at $x_0$ and the other end is $x_1$ such that $\sigma=\vec{x_0x_1}$. Let $\rho\ge 0$. We define the  hyperplane
	$$\mathscr{P}[x_0,\rho,\sigma] \ :=\ \Big\{x+\sigma \rho ~~\Big|~~ x\in\mathbb{R}^3, \ \ x\cdot \sigma=0\Big\}.$$
	\item[(C)] We suppose that the hyperplane $\mathscr{P}[x_0,\rho,\sigma]$ intersects the ball $B(x_0,R)=\Big\{x~~|~~|x-x_0|
	\le R\Big\}$ centered at $x_0$ and  radius $R$. We define the cap
	$$\mathrm{Cap}[x_0,R,\rho,\sigma]\ := \ \Big\{x\in S(x_0,R) ~~\Big| ~~ (x-x_0)\cdot \sigma \ge   \rho \Big\},$$ and $S(x_0,R)$ is the sphere $\partial B(x_0,R)$.

		\item[(D)]
	 We define the intersection between the cap  $\mathrm{Cap}[x_0,R,\rho,\sigma]$
and the hyperplane $\mathscr{P}[x_0,\rho,\sigma]$ by $\mathrm{Circle}[x_0,R,\rho,\sigma]$, which is in fact a circle, whose radius is $r[x_0,R,\rho,\sigma]=\sqrt{R^2-\rho^2}$. We denote the center of this circle by $\mathrm{CircleCenter}[x_0,R,\rho,\sigma]$. We call $\mathrm{Circle}[x_0,R,\rho,\sigma]$ the base of the cap. We denote the intersection of the half line starting from $x_0$  going through $\mathrm{CircleCenter}[x_0,R,\rho,\sigma]$ and the sphere $S(x_0,R)$ by $\mathrm{CapCenter}[x_0,R,\rho,\sigma]$. We also say that  $\mathrm{Cap}[x_0,R,\rho,\sigma]$ is the cap of chord $2r[x_0,R,\rho,\sigma]$ of $S(x_0,R)$ centered at $\mathrm{CapCenter}[x_0,R,\rho,\sigma]$.

	\item[(E)] We denote the spherical cone generated by $\mathrm{Cap}[x_0,R,\rho,\sigma]$
	$$\mathrm{Cone}[x_0,R,\rho,\sigma] \ :=\ \Big\{x\in B(x_0,R) ~~ \Big| ~~ (x-x_0)\cdot \sigma \ge  |x-x_0|\rho/R\Big\},$$
	whose volume is denoted by $\mathrm{VCone}[x_0,R,\rho].$
	\item[(F)] Let $C_1,\cdots,C_N$ be the centers of $N$ caps of chord $2r$ of the sphere $S(x_0,R)$. The area of the union of those caps is then $\mathrm{Area}\Big(\cup_{i=1}^NC_i\Big)$ and we denote
	$$\Psi[R,r,x_0,C_1,\cdots,C_N] \ : = \ \frac{\mathrm{Area}\Big(\cup_{i=1}^NC_i\Big)}{\mathrm{Area}(S(x_0,R))},$$
	where $\mathrm{Area}(S(x_0,R))=4\pi R^2$ is the area of $S(x_0,R)$. 
\end{itemize}
\end{definition}

Based on the above definitions, we can prove a ``packaging and covering'' lemma, which is crucial for our next estimates.
\begin{lemma}
	\label{Lemma:CollisionRegionSphere}
	Suppose that $\mathcal{S}$ has a non-empty interior, then $\mathrm{Coll}[\mathcal{S}]=\mathbb{R}^3\backslash\{O\}$ when $\omega(k)=|k|^2$. Suppose that the origin belongs to the interior of $\mathcal{S}$, then $\mathrm{Coll}[\mathcal{S}]=\mathbb{R}^3\backslash\{O\}$ when $\omega(k)\ne|k|^2$.
\end{lemma}
\begin{proof} First, we consider the case when $\omega(k)=|k|^2$.
	As $\mathcal{S}_0[\mathcal S]$ has a non-empty interior, there exists a ball $B(k_0,R)$ centered at $k_0$ with radius $R$ such that $B(k_0,R)\subset \mathcal{S}_0$. We then set $\mathcal{S}_0^*= B(k_0,R)$, $\mathcal{S}_{n+1}^*:=\{k\in \mathbb{R}^d ~~|~~ k=k_1+k_2-k_3; |k|^2=|k_1|^2+|k_2|^2-|k_3|^2;k_1,k_2,k_3\in\mathcal{S}_n^*;k_1+k_2\ne 0\}\backslash\{O\}$. The new collision region associated to $S_0^*$ is then 
	$$\mathcal{S}_* \ := \ \bigcup_{n=0}^\infty \mathcal{S}_n^*.$$
	It is then straightforward that $\mathcal{S}_*\subset \mathrm{Coll}[\mathcal{S}].$  
	
	Next we cover the sphere $S(k_0,R)$ with balls using a strategy  which is somewhat standard in the  area of ``packaging and covering'' (see, for instance \cite{bourgain1991convering,erdHos1964amount,fejes1985stable,toth2014regular,rogers1957note,rogers1963covering,rogers1964book}) as follows. Let $N$ be a sufficient large number to be fixed later and   $r\in (0,R/1000]$. Note that the origin is denoted by $O$ and we set  $k_0=\vec{OK}$.  We compute the mean $E\Psi[R,\sqrt{R^2-r^2},k_0,D_1,$ $\cdots,D_N]$ over all possibilities of distributing $N$ distinct point $D_1,$ $D_2,$ $\cdots,$ $D_N$ on the sphere $S(k_0,R)$ such that   $D_1,\cdots,D_N$ are the centers of $N$ caps $E_1=$ $\mathrm{Cap}[k_0,R,$ $\sqrt{R^2-r^2},(\vec{OD_1}-k_0)/|\vec{OD_1}-k_0|],$ $E_2=$ $\mathrm{Cap}[k_0,R,\sqrt{R^2-r^2},$ $(\vec{OD_2}-k_0)/|\vec{OD_2}-k_0|],$ $\cdots,$ $E_N=$ $\mathrm{Cap}[k_0,R, $ $\sqrt{R^2-r^2},$ $(\vec{OD_N}-k_0)/|\vec{OD_N}-k_0|]$ of chord $2r$ of the sphere $S(k_0,R)$. We define $q_i=\frac{\mathrm{Area}(E_i)}{\mathrm{Area}(S(k_0,R))}$. It is then clear that $q_i$ is the probability that a point $D$ on the sphere $S(k_0,R)$ belongs to the cap $E_i$. As $q_i$ is the ratio between the areas of the cap $E_i$ and the sphere $S(k_0,R)$, it is the same as  the ratio between the volumes of the spherical cone generated by $E_i$ and the ball $B(k_0,R)$, which is $q_i=\frac{\mathrm{VCone}[k_0,R,\rho]}{\frac43 \pi R^3}$ with $\rho=\sqrt{R^2-r^2}$. Therefore, the probability that $D$ does not belong to any of the caps $E_1,\cdots,E_N$ is $$(1-q_1)\cdots(1-q_N) \ = \ \Big(1-\frac{\mathrm{VCone}[k_0,R,\rho]}{\frac43 \pi R^3}\Big)^N.$$
	The expectation $\mathbb{E}\big[\Psi[R,\sqrt{R^2-r^2},k_0,D_1,\cdots,D_N]\big]$ over all possibilities of distributing $N$ distinct point $D_1,$ $D_2,$ $\cdots,$ $D_N$ can be computed by complementary
	\begin{equation}
		\label{Lemma:CollisionRegionSphere:E1}\mathbb{E}\big[\Psi[R,\sqrt{R^2-r^2},k_0,D_1,\cdots,D_N]\big] \ = \ 1- \Big(1-\frac{3\mathrm{VCone}[k_0,R,\rho]}{4 \pi R^3}\Big)^N.
	\end{equation}

	We choose $N$ such that
		\begin{equation}
		\label{Lemma:CollisionRegionSphere:E2}\Big(1-\frac{3\mathrm{VCone}[k_0,R,\rho]}{4 \pi R^3}\Big)^N<\frac{3\mathrm{VCone}[k_0,R,\rho]}{40 \pi R^3},
	\end{equation}
yielding a lower bound for   the expectation $\mathbb{E}\big[\Psi[R,\sqrt{R^2-r^2},k_0,D_1,\cdots,D_N]\big]$ over all possibilities of distributing   $D_1,$ $D_2,$ $\cdots,$ $D_N$
	\begin{equation}
	\label{Lemma:CollisionRegionSphere:E1a}\mathbb{E}\big[\Psi[R,\sqrt{R^2-r^2},k_0,D_1,\cdots,D_N]\big] \ >  \ 1- \frac{3\mathrm{VCone}[k_0,R,\rho]}{40 \pi R^3}.
\end{equation}
	Then there exists one collection $D_1,D_2,\cdots,D_N$ such that 
		\begin{equation}
		\label{Lemma:CollisionRegionSphere:E3}\Psi[R,\sqrt{R^2-r^2},k_0,D_1,\cdots,D_N] \ > \ 1- \frac{3\mathrm{VCone}[k_0,R,\rho]}{40 \pi R^3}.
	\end{equation}
	Now, keeping the centers $D_1,D_2,\cdots,D_N$  and replacing the caps $E_1,\cdots, E_N$ by caps $\mathcal E_1=$ $\mathrm{Cap}[k_0,R,$ $\sqrt{R^2-9r^2},(\vec{OD_1}-k_0)/|\vec{OD_1}-k_0|],\cdots, $ $\mathcal E_N=$ $\mathrm{Cap}[k_0,R,\sqrt{R^2-9r^2},(\vec{OD_N}-k_0)/|\vec{OD_N}-k_0|]$ whose centers are still $D_1,D_2,\cdots,D_N$   but with chord $6r$. We will show that the union of the new caps $\mathcal E_1,\cdots, \mathcal E_N$ covers the whole sphere $S(k_0,R)$. Suppose that there exists a point $D_0$ on the sphere $S(0,R)$ that does not belong to any cap in the collection $\mathcal E_1,\cdots, \mathcal E_N$. This means the cap $ E_0=\mathrm{Cap}[k_0,R,\sqrt{R^2-r^2},$ $(\vec{OD_0}-k_0)/|\vec{OD_0}-k_0|]$ with chord $2r$ does not intersect with any of the original caps $E_1,\cdots,E_N$. Therefore, the union $\cup_{i=0}^N E_i$ does not cover the full sphere $S(0,R)$, yielding
		\begin{equation}
		\label{Lemma:CollisionRegionSphere:E4} 1 \ >\ \Psi[R,\sqrt{R^2-r^2},k_0,D_1,\cdots,D_N]\ +\ \Psi[R,\sqrt{R^2-r^2},k_0,D_0], 	\end{equation}
	which, by \eqref{Lemma:CollisionRegionSphere:E3}, yields
		\begin{equation}
		\label{Lemma:CollisionRegionSphere:E5} \frac{3\mathrm{VCone}[k_0,R,\rho]}{40 \pi R^3} \ > \ \Psi[R,\sqrt{R^2-r^2},k_0,D_0] \ = \ \frac{3\mathrm{VCone}[k_0,R,\rho]}{4\pi R^3} , 	\end{equation}
	leading to a contradiction. As a result, the union of the  new caps $\mathcal E_1,\cdots, \mathcal E_N$ cover the whole sphere $S(k_0,R)$.
	
	Now, keeping the centers $D_1,D_2,\cdots,D_N$  and replacing the caps $\mathcal E_1,\cdots, \mathcal E_N$ by caps $\mathscr E_1=$ $\mathrm{Cap}[k_0,R,$ $\sqrt{R^2-36 r^2},(\vec{OD_1}-k_0)/|\vec{OD_1}-k_0|],\cdots, \mathscr E_N=$ $\mathrm{Cap}[k_0,R,$ $\sqrt{R^2-36r^2},(\vec{OD_N}-k_0)/|\vec{OD_N}-k_0|]$ whose centers are still $D_1,D_2,\cdots,D_N$   but with chord $12r$. Those caps have bases $\mathscr B_1=$ $\mathrm{Circle}[k_0,R,\sqrt{R^2-36 r^2},$ $(\vec{OD_1}-k_0)/|\vec{OD_1}-k_0|],$ $\cdots, \mathscr B_N=$ $\mathrm{Circle}[k_0,R,$ $\sqrt{R^2-36r^2},$ $(\vec{OD_N}-k_0)/|\vec{OD_N}-k_0|]$ with center $\mathscr O_1, \cdots,\mathscr O_N$. We consider the union 
		\begin{equation}
		\label{Lemma:CollisionRegionSphere:E6}
\mathscr U \ = \ \bigcup_{i=1}^N {B(\mathscr O_i,6r)} \bigcup B(k_0,R).
\end{equation}
	We will prove that the ball $B(k_0,\mathscr R(r,R))$ lies entirely inside the union $\mathscr U$ if $\mathscr{R}(r,R):=\sqrt{R^2-45r^2} + 3\sqrt 2r>R$. To see this, let us consider the two adjacent  ball $B(\mathscr O_i,6r)$ and $B(\mathscr O_j,6r)$. Suppose that the spheres  $S(\mathscr O_i,6r)$ and $S(\mathscr O_j,6r)$ intersect at a point $M$ outside the ball $B(k_0,R)$. Let $P$ be the midpoint of the segment $\mathscr O_i\mathscr O_j$.  From our construction $\mathscr O_i$ lies  in the ball $B(\mathscr O_j,6r)$ and  $\mathscr O_j$ lies  in the ball $B(\mathscr O_i,6r)$. Therefore the angle between the two vectors $\vec{M\mathscr O_i} $ and $\vec{M\mathscr O_j} $ is smaller than $\frac{\pi}{2}$, which implies that $MP\ge \frac{M\mathscr O_j}{\sqrt 2}$ yielding $MP\ge 3\sqrt 2r$. Since $K\mathscr O_1=K\mathscr O_2=\sqrt{R^2-36r^2}$, we obtain $KP=\sqrt{K\mathscr O_1^2-P\mathscr O_1^2}\ge \sqrt{R^2-36r^2-9r^2} =\sqrt{R^2-45r^2}$. Therefore  $ KM=KP+PM\ge \sqrt{R^2-45r^2}+ 3\sqrt 2r=\mathscr{R}(r,R)$. As a result,  the ball $B(k_0,\mathscr R(r,R))$ lies entirely inside the union $\mathscr U$.

Let $A,B$ be two points belonging to the base $\mathscr{B}_i$ such that $A,B$ and the center $\mathscr O_i$ are on the same line. 
We set $\vec{OA}=k_1$ and $\vec{OB}=k_2$.
Let $k$ be a point in the set $\mathcal{S}_1^*$ such that there exists $k_3\in \mathcal{S}_0^*$ satisfying  $k=k_1+k_2-k_3$ and $ |k|^2=|k_1|^2+|k_2|^2-|k_3|^2$. A straightforward computation gives $(k-k_1)\cdot (k-k_2)=0$. Therefore, the sphere $S(\mathscr O_i,r)$ centered at $\mathscr O_i$ with radius $r$ belongs to the set $\mathcal{S}_1^*$, yielding that the whole ball ${B(\mathscr O_i,r)}$ centered at $\mathscr O_i$ with radius $r$  also belong to the set $\mathcal{S}_1^*\cup  \{O\}$. Therefore the union $\mathscr U$ and the ball $B(k_0,\mathscr R(r,R))$ also lie entirely inside $\mathcal{S}_1^*\cup  \{O\}$. 

As a result, since $B(k_0,R)\subset \mathcal{S}_0^*$, we obtain 
the ball $B(k_0,\mathscr R(r,R))$ also lie entirely inside $\mathcal{S}_1^*\cup  \{O\}$.
 Repeating iteratively this argument for $\mathcal{S}_2^*$, $\mathcal{S}_3^*$, $\mathcal{S}_4^*$, etc. we deduce that $\mathcal{S}_*=\mathbb{R}^3\backslash\{O\}$ and hence $\mathrm{Coll}[\mathcal{S}]=\mathbb{R}^3\backslash   \{O\}$. 
\\

Now, we consider the case when $\omega_k$ is not $|k|^2$. The proof can be done in a similar manner as  the case of $\gamma=2$. 	As the origin belongs to the interior of $\mathcal{S}_0\cup\{O\}$, there exists a ball $B(O,R)$  such that $B(O,R)\subset \mathcal{S}_0\cup\{O\}$. We also set $\mathcal{S}_0^o= B(O,R)$, $\mathcal{S}_{n+1}^o:=\{k\in \mathbb{R}^d ~~|~~ k=k_1+k_2-k_3; \omega_k=\omega_{k_1}+\omega_{k_2}-\omega_{k_3};k_1+k_2\ne0;k_1,k_2,k_3\in\mathcal{S}_n^o\}$. The  collision region associated to $S_0^o$ is then 
$$\mathcal{S}_o \ := \ \bigcup_{n=0}^\infty \mathcal{S}_n^o,$$
 and hence $\mathcal{S}_o\subset \mathrm{Coll}[\mathcal{S}]\cup   \{O\}.$  Let $X,Y$ be two points on the sphere $S(O,R)$. 
 We set $\vec{OX}=k_1,$ $\vec{OY}=k_2$.
 Let $k$ be a point in the set $\mathcal{S}_1$ such that there exists $k_3\in \mathcal{S}_0^*$ satisfying  $k=k_1+k_2-k_3$ and $ \omega(|k|)=\omega(|k_1|)+\omega(|k_2|)-\omega(|k_3|)$. We set $k=\vec{OW}$ and $k_3=\vec{OZ}$ and we denote by $U$ the common midpoint of the 2 segments $WZ$ and $XY$.  Let us consider the case when $O,U,Z,W$ are on the same line. In this case, we set $\vec{OW}=(1+s)\vec{OU}$ and $\vec{OZ}=(1-s)\vec{OU}$. We obtain the equation $\omega(|1+s||\vec{OU}|)+\omega(|1-s||\vec{OU}|) =2\omega(R),$ yielding $\omega(|1+s|\kappa)+\omega(|1-s|\kappa) =2\omega(R),$ with $\kappa=|\vec{OU}|$. Let us consider the function $\digamma
 (s)=\omega(|1+s|\kappa)+\omega(|1-s|\kappa)$ with $s\in [1,2]$. Since  $\digamma
 '(s)=(s+1)\omega'(|1+s|\kappa)+(s-1)\omega'(|1-s|\kappa)\ge0$, we find $\omega(2\kappa)\le \omega(|1+s|\kappa)+\omega(|1-s|\kappa)\le \omega(3\kappa)+\omega(\kappa)$ for $s\in [1,2]$. We denote   by $2\alpha$ the angle between $k_1$ and $k_2$, then $\kappa=R\cos\alpha $. We find $\omega(2R\cos\alpha)\le \omega(|1+s|R\cos\alpha)+\omega((s-1)R\cos\alpha)\le \omega(3R\cos\alpha)+\omega(R\cos\alpha)$ for $s\in [1,2]$. Let us consider the case when $\alpha=\alpha_o=\frac{\pi}{3}$, then $\omega(2R\cos\alpha_o)=\omega(R)<2\omega(R)<\omega(3R/2)+\omega(R/2)=\omega(3R\cos\alpha_o)+\omega(R\cos\alpha_o)$. Therefore,  there exists a solution $s_0\in (1,2)$ to the equation $\digamma(s_0)=2\omega(R)$. We can then choose $k=(1+s_0)\vec{OU}$ and hence $(1+s_0)\vec{OU}$ also belongs to the set  $\mathcal{S}_1^o$.
  Therefore  the ball $B(O,(1+s_0)\kappa)$ also lie entirely inside $\mathcal{S}_1^o$.  
  
  As a result, starting with $B(O,R)\subset \mathcal{S}_0$, we can prove $B(O,(1+s_0)\kappa)$ also lie entirely inside $\mathcal{S}_1^o$. 
  Repeating inductively this argument for $\mathcal{S}_2^o$, $\mathcal{S}_3^o$, $\mathcal{S}_4^o$, etc. we deduce that $\mathcal{S}_o=\mathbb{R}^3$  and hence $\mathrm{Coll}[\mathcal{S}]=\mathbb{R}^3\backslash  \{O\}$. 
 \end{proof}

\subsection{The support of the solutions}
Using  the   ``packaging and covering'' lemma, we can now prove the following proposition.

\begin{proposition}
		\label{Lemma:Support} Let $f$ be a mild solution in the sense of \eqref{4wavemild}. Suppose that the support of the initial condition $\mathcal S=\mathrm{supp} f_0$ satisfies the assumption of Lemma \ref{Lemma:CollisionRegionSphere}. Then,
	for any $k\in \mathbb{R}^3,$ there exists a   time $T_0>0$, such that $k$ belongs to the support of $f(T_0)$. 
\end{proposition}

\begin{proof}
	
	Let $k_0$ be any point in $\mathbb{R}^3$ and $T>0$, by Lemma \ref{Lemma:CollisionRegionSphere}, there exists a number $M\in\mathbb{N}$ such that $k_0\in \mathcal S_M[\mathrm{supp} f_0]$. We set $k_{M,0}=k_0$ and there exists a set $\{k_{n,0},\cdots, k_{n,M-n}\}\subset \mathcal S_n[\mathrm{supp} f_0]\backslash \mathcal S_{n+1}[\mathrm{supp} f_0] $ with $n=0,\cdots,M$  that satisfies
$$k_{n,\Psi_n} \ = \ k_{n-1,\Psi_n} \ +  \ k_{n-1,\Psi_n+1} \ - \ k_{n-1,\Psi_n+2},$$
$$\omega(k_{n,\Psi_n}) \ = \ \omega(k_{n-1,\Psi_n}) \ +  \ \omega(k_{n-1,\Psi_n+1}) \ - \ \omega(k_{n-1,\Psi_n+2}),$$
$$ \ \ \ k_{n-1,i} \ =\ k_{n,i} \mbox{ with } 0\le i<\Psi_n; \ \ \ k_{n-1,i+2} \ =\ k_{n,i} \mbox{ with }  i>\Psi_n,$$
for $\Psi_n\in\{0,\cdots,M-n\}$, and $k_{n-1,\Psi_n}   +    k_{n-1,\Psi_n+1}\ne0$. In the proof of this lemma, we define $\mathcal B(k,r)$ to be
\begin{equation}\label{Ball}
	\mathcal 	B(k,r) = \Big\{k'\in\mathbb{R}^3 ~~ |~~\big||k'|-|k|\big|\le r \Big\}.
\end{equation}
We also define $\varpi_{\mathcal 	B(k,r)}(k')$ to be the smooth pump function that satisfies $\varpi_{\mathcal 	B(k,r)}(|k'|)=1$ when $|k'|\in \mathcal 	B(k,999r/1000)$
and $\varpi_{\mathcal 	B(k,r)}(|k'|)=0$ when $|k'|\notin \mathcal 	B(k,r)$. Therefore, for a given small sufficiently constant $\rho>0$, to obtain the conclusion of the proposition, we will prove that there exists a time $T_0$ such that 
\begin{equation}\int_{\mathbb{R}^d}\mathrm d k f(T_0,k) \varpi_{\mathcal 	B(k_0,\rho)}(k)\ > \ 0.
\end{equation}

 There exist sufficiently small positive numbers $\rho_{n}>0$ such that $\rho_{M}=\rho$ and  for $k+k_1=k_2+k_3$ and $\omega+\omega_1=\omega_2+\omega_3$, we have  
\begin{equation}\label{varpi}\varpi_{\mathcal B(k_{n,\Psi_n},\rho_n)}(k)\ge \varpi_{\mathcal B(k_{n-1,\Psi_n+2},\rho_{n-1})}(k_1)   \varpi_{\mathcal B(k_{n-1,\Psi_n+1},\rho_{n-1})}(k_2)  \varpi_{\mathcal B(k_{n-1,\Psi_n},\rho_{n-1})}(k_3),\end{equation}
which can be proved as follows. Let $\bar{k}_1\in \mathrm{support} \Big[\varpi_{\mathcal B(k_{n-1,\Psi_n+2},\rho_{n-1})}(k_1) \Big]= {\mathcal B(k_{n-1,\Psi_n+2},\rho_{n-1})}$, $\bar{k}_2\in  \mathrm{support} \Big[\varpi_{\mathcal B(k_{n-1,\Psi_n+1},\rho_{n-1})}(k_2) \Big]= \mathcal B(k_{n-1,\Psi_n+1},\rho_{n-1}) $,  $\bar{k}_3\in  \mathrm{support} \Big[\varpi_{\mathcal B(k_{n-1,\Psi_n},\rho_{n-1})}(k_3) \Big]= \mathcal B(k_{n-1,\Psi_n},\rho_{n-1})$ and  $\bar k+\bar k_1=\bar k_2+\bar k_3$ and $\omega(\bar k)+\omega(\bar k_1)=\omega(\bar k_2)+\omega(\bar k_3)$. For $\rho_{n-1}$ being sufficiently small $|\omega(\bar{k}_1)-\omega(k_{n-1,\Psi_n+2})|\le 3\rho_{n-1}|\omega'(k_{n-1,\Psi_n+2})|$, $|\omega(\bar{k}_2)-\omega(k_{n-1,\Psi_n+1})|\le 3 \rho_{n-1}|\omega'(k_{n-1,\Psi_n+1})|$, $|\omega(\bar{k}_3)-\omega(k_{n-1,\Psi_n})|\le 3\rho_{n-1}|\omega'(k_{n-1,\Psi_n})|$, yielding $|\omega(\bar{k})-\omega(k_{n,\Psi_n})|\le 3\rho_{n-1}[|\omega'(k_{n-1,\Psi_n+2})|+|\omega'(k_{n-1,\Psi_n+1})|+|\omega'(k_{n-1,\Psi_n})|]$. We then find $|\bar{k}-k_{n,\Psi_n}||\omega'(k_{n,\Psi_n})|/3\le 3\rho_{n-1}[|\omega'(k_{n-1,\Psi_n+2})|+|\omega'(k_{n-1,\Psi_n+1})|+|\omega'(k_{n-1,\Psi_n})|]$, yielding \eqref{varpi} when $\rho_n,\rho_{n-1}$ are small.

We also suppose that $3\rho_n<|k_{n,\Psi_n}|$.
Using $\varpi_{\mathcal B(k_{n,\Psi_n},\rho_n)}(k)$ as a test function, we find 
\begin{equation}\label{Lemma:Support:E1}
	\begin{aligned}
		&  \int_{\mathbb{R}^3}\mathrm{d}k f(T,k) \varpi_{B(k_{n,\Psi_n},\rho_n)}(k) - \int_{\mathbb{R}^3}\mathrm{d}k f(0,k) \varpi_{B(k_{n,\Psi_n},\rho_n)}(k)\\
		\ =\ & \int_{0}^{T}\mathrm{d}t\iiiint_{\mathbb{R}^{3\times4}}\mathrm{d}k_1\,\mathrm{d}k_2\,\mathrm{d}k_3\mathrm{d}k \delta(k+k_1-k_2-k_3)\delta(\omega + \omega_1 -\omega_2 - \omega_3)f_1f_2f_3 \\
		&\times \left[ \varpi_{\mathcal B(k_{n,\Psi_n},\rho_n)}(k) \ + \ \varpi_{\mathcal B(k_{n,\Psi_n},\rho_n)}(k_1) \ - \ \varpi_{\mathcal B(k_{n,\Psi_n},\rho_n)}(k_2) \ - \ \varpi_{\mathcal B(k_{n,\Psi_n},\rho_n)}(k_3)
		\right].
	\end{aligned}
\end{equation}

We  will estimate
\begin{equation}\label{Lemma:Support:E2}
	\begin{aligned}
		& \iiiint_{\mathbb{R}^{3\times4}}\mathrm{d}k_1\,\mathrm{d}k_2\,\mathrm{d}k_3\mathrm{d}k \delta(k+k_1-k_2-k_3)\delta(\omega + \omega_1 -\omega_2 - \omega_3)f_2f_3f_1 \\
		&\times \left[\varpi_{\mathcal B(k_{n,\Psi_n},\rho_n)}(k) \ + \ \varpi_{\mathcal B(k_{n,\Psi_n},\rho_n)}(k_1) \ - \ \varpi_{\mathcal B(k_{n,\Psi_n},\rho_n)}(k_2) \ - \ \varpi_{\mathcal B(k_{n,\Psi_n},\rho_n)}(k_3)
		\right].
	\end{aligned}
\end{equation}
To this end, we  divide the rest of the proof into smaller steps.

{\it Step 1: The first a priori estimate.} 

We first estimate, for $0<\zeta<\min\{1,|k_{n,\Psi_n}|-3\rho_n,|k_{n-1,\Psi_{n-1}}|-3\rho_{n-1}\}/1000$ being a sufficiently small constant
\begin{equation}\label{Lemma:Support:E3a}
	\begin{aligned}
		&	\iiiint_{\mathbb{R}^{3\times4}}\mathrm{d}k_1\,\mathrm{d}k_2\,\mathrm{d}k_3\mathrm{d}k \delta(k+k_1-k_2-k_3)\delta(\omega + \omega_1 -\omega_2 - \omega_3)f_2f_3f_1 \\
		&\times \Big[   \varpi_{\mathcal B(k_{n,\Psi_n},\rho_n)}(k_2) \ + \ \varpi_{\mathcal B(k_{n,\Psi_n},\rho_n)}(k_3) \\
		& \ - \ \varpi_{\mathcal B(k_{n,\Psi_n},\rho_n)}(k_1) \ - \ \varpi_{\mathcal B(k_{n,\Psi_n},\rho_n)}(k)\chi_{\{|k_1|,|k_3|\le \zeta\}\cup\{|k_1|,|k_2|\le \zeta\}}	\Big].
	\end{aligned}
\end{equation}
Using the same argument with \eqref{Lemma:TestFunction:E4},
we deduce
\begin{equation}\label{Lemma:Support:E3b}
	\begin{aligned}
		&	\iiiint_{\mathbb{R}^{3\times4}}\mathrm{d}k_1\,\mathrm{d}k_2\,\mathrm{d}k_3\mathrm{d}k \delta(k+k_1-k_2-k_3)\delta(\omega + \omega_1 -\omega_2 - \omega_3)f_2f_3f_1 \\
		&\times \Big[   \varpi_{\mathcal B(k_{n,\Psi_n},\rho_n)}(k_2) \ + \ \varpi_{\mathcal B(k_{n,\Psi_n},\rho_n)}(k_3)\\
		& \ - \ \varpi_{\mathcal B(k_{n,\Psi_n},\rho_n)}(k_1) \ - \ \varpi_{\mathcal B(k_{n,\Psi_n},\rho_n)}(k)\chi_{\{|k_1|,|k_3|\le \zeta\}\cup\{|k_1|,|k_2|\le \zeta\}}	\Big]\\
		\le\	&	 C_Q\iiint_{\mathbb{R}_+^{3}}\mathrm{d}|k_1|\,\mathrm{d}|k_2|\,\mathrm{d}|k_3|{|k_1||k_2||k_3|\min\{|k_1|,|k_2|,|k_3|,|k|\}}\frac{|k|}{\omega'(|k|)}  f_2f_3f_1\\
		&\times \Big[   \varpi_{\mathcal B(k_{n,\Psi_n},\rho_n)}(k_2) \ + \ \varpi_{\mathcal B(k_{n,\Psi_n},\rho_n)}(k_3)\\
		& \ - \ \varpi_{\mathcal B(k_{n,\Psi_n},\rho_n)}(k_1) \ - \ \varpi_{\mathcal B(k_{n,\Psi_n},\rho_n)}(k)\chi_{\{|k_1|,|k_3|\le \zeta\}\cup\{|k_1|,|k_2|\le \zeta\}}	\Big] \\
		\lesssim\	&	 \iiint_{\mathbb{R}_+^{3}}\mathrm{d}|k_1|\,\mathrm{d}|k_2|\,\mathrm{d}|k_3|\frac{\min\{|k|,|k_1|,|k_2|,|k_3|\}}{|k_3||k_1||k_2|}  f_2|k_2|^2f_3|k_3|^2f_1|k_1|^2\frac{|k|}{\omega'(|k|)} \\
		&\times \Big[   \varpi_{\mathcal B(k_{n,\Psi_n},\rho_n)}(k_2) \ + \ \varpi_{\mathcal B(k_{n,\Psi_n},\rho_n)}(k_3)\\
		& \ - \ \varpi_{\mathcal B(k_{n,\Psi_n},\rho_n)}(k_1) \ - \ \varpi_{\mathcal B(k_{n,\Psi_n},\rho_n)}(k)\chi_{\{|k_1|,|k_3|\le \zeta\}\cup\{|k_1|,|k_2|\le \zeta\}}	\Big],
	\end{aligned}
\end{equation}
where $|k|$ is the variable associated to $\omega=\omega_2+\omega_1-\omega_3$. 

Let us consider the following $4$ cases, for $k+k_1=k_2+k_3$ and $\omega + \omega_1 =\omega_2 + \omega_3$:

{\it Case (a):}  If $|k_1|=\min\{|k|,|k_1|,|k_2|,|k_3|\}$,  we consider also $2$ cases. If $|k_3|>\zeta$ and $|k_2|>\zeta$, then $|k|\ge \zeta$ and we bound
\begin{equation}\label{Lemma:Support:E3d}
	\begin{aligned}
		&	\iiiint_{\mathbb{R}^{3\times4}}\mathrm{d}k_1\,\mathrm{d}k_2\,\mathrm{d}k_3\mathrm{d}k \delta(k+k_1-k_2-k_3)\delta(\omega + \omega_1 -\omega_2 - \omega_3)f_2f_3f_1 \\
		&\times   \chi_{\{|k_1|=\min\{|k|,|k_1|,|k_2|,|k_3|\}\}}\chi_{\{|k_3|,|k_2|>\zeta\}}\\
		&\times \Big[   \varpi_{\mathcal B(k_{n,\Psi_n},\rho_n)}(k_2) \ + \ \varpi_{\mathcal B(k_{n,\Psi_n},\rho_n)}(k_3) \ - \ \varpi_{\mathcal B(k_{n,\Psi_n},\rho_n)}(k_1) 
		\\ &	 \ - \ \varpi_{\mathcal B(k_{n,\Psi_n},\rho_n)}(k)\chi_{\{|k_1|,|k_3|\le \zeta\}\cup\{|k_1|,|k_2|\le \zeta\}}\Big]\\
		\lesssim\	&	 \iiint_{\mathbb{R}_+^{3}}\mathrm{d}|k_1|\,\mathrm{d}|k_2|\,\mathrm{d}|k_3|\frac{\min\{|k|,|k_1|,|k_2|,|k_3|\}}{|k_3||k_1||k_2|}  f_2|k_2|^2f_3|k_3|^2f_1|k_1|^2\frac{|k|}{\omega'(|k|)} \chi_{\{|k_3|,|k_2|>\zeta\}}\\
		&\times \Big[   \varpi_{\mathcal B(k_{n,\Psi_n},\rho_n)}(k_2) \ + \ \varpi_{\mathcal B(k_{n,\Psi_n},\rho_n)}(k_3)\Big]\\
		\lesssim\	&	 \iiint_{\mathbb{R}_+^{3}}\mathrm{d}|k_1|\,\mathrm{d}|k_2|\,\mathrm{d}|k_3|\frac{\min\{|k|,|k_1|,|k_2|,|k_3|\}}{|k_3||k_1||k_2|}  f_2|k_2|^2f_3|k_3|^2f_1|k_1|^2|k|^\iota\chi_{\{|k_3|,|k_2|>\zeta\}}\\
		&\times \Big[   \varpi_{\mathcal B(k_{n,\Psi_n},\rho_n)}(k_2) \ + \ \varpi_{\mathcal B(k_{n,\Psi_n},\rho_n)}(k_3)\Big]\\
		\lesssim\	&	 \iiint_{\mathbb{R}_+^{3}}\mathrm{d}|k_1|\,\mathrm{d}|k_2|\,\mathrm{d}|k_3|\frac{\min\{|k|,|k_1|,|k_2|,|k_3|\}}{|k_3||k_1||k_2|}  f_2|k_2|^2f_3|k_3|^2f_1|k_1|^2(|k_1|^\iota+|k_2|^\iota+|k_3|^\iota)\chi_{\{|k_3|,|k_2|>\zeta\}}\\
		&\times \Big[   \varpi_{\mathcal B(k_{n,\Psi_n},\rho_n)}(k_2) \ + \ \varpi_{\mathcal B(k_{n,\Psi_n},\rho_n)}(k_3)\Big]\\
		\lesssim\	&		\frac{\mathfrak{C}_2'}{\zeta}  \int_{\mathbb{R}^3}\mathrm{d}k f(t,k) \chi_{\mathcal B(k_{n,\Psi_n},\rho_n)}(k),
	\end{aligned}
\end{equation}
for some universal constant $	\mathfrak{C}_2' >0$ independent of $\zeta$, where we have used the boundedness of $\int_{\mathbb{R}^3}\mathrm{d}k f\omega$ and $\int_{\mathbb{R}^3}\mathrm{d}k f$.

Next, we will consider the case when $|k_2|\le \zeta$ or $|k_3|\le \zeta$. Since $k_2,k_3$ are symmetric, we only consider the case when
$|k_3|\le \zeta$, then $|k_1|\le \zeta$ and $|k_2|\le |k|$,  we bound
\begin{equation}\label{Lemma:Support:E3e}
	\begin{aligned}
		&	\iiiint_{\mathbb{R}^{3\times4}}\mathrm{d}k_1\,\mathrm{d}k_2\,\mathrm{d}k_3\mathrm{d}k \delta(k+k_1-k_2-k_3)\delta(\omega + \omega_1 -\omega_2 - \omega_3)f_2f_3f_1    \chi_{\mathcal B(k_{n,\Psi_n},\rho_n)}(k_2)  \\
		&\times   \chi_{\{|k_1|=\min\{|k|,|k_1|,|k_2|,|k_3|\}\}}\chi_{\{|k_3|,|k_1|\le \zeta\}}\\
		&\times \Big[   \varpi_{\mathcal B(k_{n,\Psi_n},\rho_n)}(k_2) \ + \ \varpi_{\mathcal B(k_{n,\Psi_n},\rho_n)}(k_3) \ - \ \varpi_{\mathcal B(k_{n,\Psi_n},\rho_n)}(k_1)\\ 	&  \ - \ \varpi_{\mathcal B(k_{n,\Psi_n},\rho_n)}(k)\chi_{\{|k_1|,|k_3|\le \zeta\}\cup\{|k_1|,|k_2|\le \zeta\}}\Big]\\
		\le & \iiint_{\mathbb{R}_+^{3}}\mathrm{d}|k_1|\,\mathrm{d}|k_2|\,\mathrm{d}|k_3|\frac{\min\{|k|,|k_1|,|k_2|,|k_3|\}}{|k_3||k_1||k_2|}  f_2|k_2|^2f_3|k_3|^2f_1|k_1|^2\frac{|k|}{\omega'(|k|)} \\ 
		&\times   \chi_{\{|k_1|=\min\{|k|,|k_1|,|k_2|,|k_3|\}\}}\chi_{\{|k_3|,|k_1|\le \zeta\}}\\
		&\times \left[   \varpi_{\mathcal B(k_{n,\Psi_n},\rho_n)}(k_2) 	 \ - \ \varpi_{\mathcal B(k_{n,\Psi_n},\rho_n)}(k)\right]\\
		\lesssim & \iiint_{\mathbb{R}_+^{3}}\mathrm{d}|k_1|\,\mathrm{d}|k_2|\,\mathrm{d}|k_3|   \chi_{\{|k_1|=\min\{|k|,|k_1|,|k_2|,|k_3|\}\}}\chi_{\{|k_3|,|k_1|\le \zeta\}}  f_2|k_2|^2f_3|k_3|^2f_1|k_1|^2\frac{|k|}{\omega'(|k|)} \\
		&\times\frac{\varpi_{\mathcal B(k_{n,\Psi_n},\rho_n)}(k_2) 	 \ - \ \varpi_{\mathcal B(k_{n,\Psi_n},\rho_n)}(k)}{|k_3||k_2|}\\
		\lesssim & \iiint_{\mathbb{R}_+^{3}}\mathrm{d}|k_1|\,\mathrm{d}|k_2|\,\mathrm{d}|k_3|   \chi_{\{|k_1|=\min\{|k|,|k_1|,|k_2|,|k_3|\}\}}\chi_{\{|k_3|,|k_1|\le \zeta\}}  f_2|k_2|^2f_3|k_3|^2f_1|k_1|^2\\
		&\times\frac{-|k_2| 	 \ + \ |k|} {|k_3|}\frac{\Big|\varpi'_{\mathcal B(k_{n,\Psi_n},\rho_n)}(\xi(|k_2|,|k|))\Big|}{|k_2|}(|k_1|^\iota+|k_2|^\iota+|k_3|^\iota).
	\end{aligned}
\end{equation}
We have $0\le (|k|-|k_2|)\omega'(\sigma)=\omega-\omega_2=-\omega_1+\omega_3\le  C_\omega'|k_3|^{\alpha'}$. Since $\sigma>|k_{n,\Psi_n}|-3\rho_n$, we  have $0\le |k|-|k_2|\le  C_\omega''|k_3|^{\alpha'}\le  \zeta^{\alpha'-1}C_\omega''|k_3|$, yielding $0\le \frac{-|k_2| 	 \ + \ |k|} {|k_3|}\le C_\omega'' \zeta^{\alpha'-1}$. If $|k_2|\le [|k_{n,\Psi_n}|-3\rho_n]/10$, then $|k|\le [|k_{n,\Psi_n}|-3\rho_n]/2$, yielding $ \left[   \varpi_{\mathcal B(k_{n,\Psi_n},\rho_n)}(k_2) 	 \ - \ \varpi_{\mathcal B(k_{n,\Psi_n},\rho_n)}(k)\right]=0$. When $|k_2|>[|k_{n,\Psi_n}|-3\rho_n]/10$, since $\Big|\varpi'_{\mathcal B(k_{n,\Psi_n},\rho_n)}(\xi(|k_2|,|k|))\Big|$ is bounded, then $\frac{\Big|\varpi'_{\mathcal B(k_{n,\Psi_n},\rho_n)}(\xi(|k_2|,|k|))\Big|}{|k_2|}$  is bounded, we deduce

\begin{equation}\label{Lemma:Support:E3ee}
	\begin{aligned}
		&	\iiiint_{\mathbb{R}^{3\times4}}\mathrm{d}k_1\,\mathrm{d}k_2\,\mathrm{d}k_3\mathrm{d}k \delta(k+k_1-k_2-k_3)\delta(\omega + \omega_1 -\omega_2 - \omega_3)f_2f_3f_1    \chi_{\mathcal B(k_{n,\Psi_n},\rho_n)}(k_2)  \\
		&\times   \chi_{\{|k_2|=\min\{|k|,|k_1|,|k_2|,|k_3|\}\}}\chi_{\{|k_3|,|k_1|\le \zeta\}}\\
		&\times \Big[   \varpi_{\mathcal B(k_{n,\Psi_n},\rho_n)}(k_2) \ + \ \varpi_{\mathcal B(k_{n,\Psi_n},\rho_n)}(k_3) \ - \ \varpi_{\mathcal B(k_{n,\Psi_n},\rho_n)}(k_1)\\ & 	 \ - \ \varpi_{\mathcal B(k_{n,\Psi_n},\rho_n)}(k)\chi_{\{|k_1|,|k_3|\le c\}\cup\{|k_1|,|k_2|\le \zeta\}}\Big]\\
		\lesssim\	&		\mathfrak{C}_2' \zeta^{\alpha'-1},
	\end{aligned}
\end{equation}
for some universal constant $	\mathfrak{C}_2' >0$ independent of $\zeta$.

{\it Case (b):} If $|k_2|=\min\{|k|,|k_1|,|k_2|,|k_3|\}$, then on the support of $ \varpi_{\mathcal B(k_{n,\Psi_n},\rho_n)}(k_2) $, we always have $|k_1|,|k_2|,|k_3|>|k_{n,\Psi_n}|-2\rho_n>0$ and we bound
\begin{equation}\label{Lemma:Support:E3c}
	\begin{aligned}
		&	\iiiint_{\mathbb{R}^{3\times4}}\mathrm{d}k_1\,\mathrm{d}k_2\,\mathrm{d}k_3\mathrm{d}k \delta(k+k_1-k_2-k_3)\delta(\omega + \omega_1 -\omega_2 - \omega_3)f_2f_3f_1 \chi_{\{|k_2|=\min\{|k|,|k_1|,|k_2|,|k_3|\}\}}\\
		&\times \left[   \varpi_{\mathcal B(k_{n,\Psi_n},\rho_n)}(k_2) \right]\
		\lesssim\			\mathfrak{C}_2'  \int_{\mathbb{R}^3}\mathrm{d}k f(t,k) \varpi_{\mathcal B(k_{n,\Psi_n},\rho_n)}(k),
	\end{aligned}
\end{equation}
for some universal constant $	\mathfrak{C}_2' >0$. We consider $2$ cases. If $|k_1|>\zeta$, we bound
\begin{equation}\label{Lemma:Support:E3c:1}
	\begin{aligned}
		&	\iiiint_{\mathbb{R}^{3\times4}}\mathrm{d}k_1\,\mathrm{d}k_2\,\mathrm{d}k_3\mathrm{d}k \delta(k+k_1-k_2-k_3)\delta(\omega + \omega_1 -\omega_2 - \omega_3)f_2f_3f_1 \chi_{\{|k_2|=\min\{|k|,|k_1|,|k_2|,|k_3|\}\}}\\
		&\times \left[   \varpi_{\mathcal B(k_{n,\Psi_n},\rho_n)}(k_3)  \right]\chi_{\{|k_1|>\zeta\}}\
		\lesssim\	 		\frac{\mathfrak{C}_2'}{\zeta}  \int_{\mathbb{R}^3}\mathrm{d}k f(t,k) \varpi_{\mathcal B(k_{n,\Psi_n},\rho_n)}(k).
	\end{aligned}
\end{equation}
Now if $|k_1|\le \zeta$, we bound
\begin{equation}\label{Lemma:Support:E3c:2}
	\begin{aligned}
		&	\iiiint_{\mathbb{R}^{3\times4}}\mathrm{d}k_1\,\mathrm{d}k_2\,\mathrm{d}k_3\mathrm{d}k \delta(k+k_1-k_2-k_3)\delta(\omega + \omega_1 -\omega_2 - \omega_3)f_2f_3f_1 \chi_{\{|k_2|=\min\{|k|,|k_1|,|k_2|,|k_3|\}\}}\\
		&\times \left[   \varpi_{\mathcal B(k_{n,\Psi_n},\rho_n)}(k_3) \ - \ \varpi_{\mathcal B(k_{n,\Psi_n},\rho_n)}(k)\chi_{\{|k_1|,|k_3|\le \zeta\}\cup\{|k_1|,|k_2|\le \zeta\}} \right]\chi_{\{|k_1|\le \zeta\}}\\
		\lesssim &		\iiiint_{\mathbb{R}^{3\times4}}\mathrm{d}k_1\,\mathrm{d}k_2\,\mathrm{d}k_3\mathrm{d}k \delta(k+k_1-k_2-k_3)\delta(\omega + \omega_1 -\omega_2 - \omega_3)f_2f_3f_1 \chi_{\{|k_2|=\min\{|k|,|k_1|,|k_2|,|k_3|\}\}}\\
		&\times \chi_{\{|k_1|,|k_2|\le \zeta\}} \left[   \varpi_{\mathcal B(k_{n,\Psi_n},\rho_n)}(k_3) \ - \ \varpi_{\mathcal B(k_{n,\Psi_n},\rho_n)}(k)\right]\\
		\lesssim &		\iiiint_{\mathbb{R}^{3\times4}}\mathrm{d}k_1\,\mathrm{d}k_2\,\mathrm{d}k_3\mathrm{d}k \delta(k+k_1-k_2-k_3)\delta(\omega + \omega_1 -\omega_2 - \omega_3)f_2f_3f_1 \chi_{\{|k_2|=\min\{|k|,|k_1|,|k_2|,|k_3|\}\}}\\
		&\times \chi_{\{|k_1|,|k_2|\le \zeta\}} \left[   |\varpi'_{\mathcal B(k_{n,\Psi_n},\rho_n)}(\xi(|k_3|,|k|))| [|k_3|-|k|]\right].
	\end{aligned}
\end{equation}

We have $0\le (|k_3|-|k|)\omega'(\sigma)=\omega_3-\omega=\omega_1-\omega_2\le  C_\omega'|k_1|^{\alpha'}$. Since $\sigma>|k_{n,\Psi_n}|-3\rho_n$, we  have $0\le |k_3|-|k|\le  C_\omega''|k_1|^{\alpha'}\le  C_\omega''|k_1|\zeta^{\alpha'-1}$. By a similar argument with \eqref{Lemma:Support:E3e}, we bound 

\begin{equation}\label{Lemma:Support:E3c:3}
	\begin{aligned}
		&	\iiiint_{\mathbb{R}^{3\times4}}\mathrm{d}k_1\,\mathrm{d}k_2\,\mathrm{d}k_3\mathrm{d}k \delta(k+k_1-k_2-k_3)\delta(\omega + \omega_1 -\omega_2 - \omega_3)f_2f_3f_1 \chi_{\{|k_2|=\min\{|k|,|k_1|,|k_2|,|k_3|\}\}}\\
		&\times \left[   \varpi_{\mathcal B(k_{n,\Psi_n},\rho_n)}(k_3) \ - \ \varpi_{\mathcal B(k_{n,\Psi_n},\rho_n)}(k)\chi_{\{|k_1|,|k_3|\le \zeta\}\cup\{|k_1|,|k_2|\le \zeta\}} \right]\chi_{\{|k_1|\le \zeta\}}\\
		\lesssim\	&		\mathfrak{C}_2' \zeta^{\alpha'-1},
	\end{aligned}
\end{equation}
for some universal constant $	\mathfrak{C}_2' >0$.

{\it Case (c):}  If $|k_3|=\min\{|k|,|k_1|,|k_2|,|k_3|\}$,  then since $k_2,k_3$ are symmetric, the same argument  in Case (b) can be repeated. Finally, we also arrive at 
\begin{equation}\label{Lemma:Support:E3f}
	\begin{aligned}
		&	\iiiint_{\mathbb{R}^{3\times4}}\mathrm{d}k_1\,\mathrm{d}k_2\,\mathrm{d}k_3\mathrm{d}k \delta(k+k_1-k_2-k_3)\delta(\omega + \omega_1 -\omega_2 - \omega_3)f_2f_3f_1 \chi_{\{|k_3|=\min\{|k|,|k_1|,|k_2|,|k_3|\}\}}\\
		&\times \left[   \varpi_{\mathcal B(k_{n,\Psi_n},\rho_n)}(k_2) \ + \ \varpi_{\mathcal B(k_{n,\Psi_n},\rho_n)}(k_3) \ - \ \varpi_{\mathcal B(k_{n,\Psi_n},\rho_n)}(k_1) 	 \ - \ \varpi_{\mathcal B(k_{n,\Psi_n},\rho_n)}(k)\chi_{\{|k_1|,|k_3|\le  \zeta\}}\right]\\
		\lesssim\	&		\frac{\mathfrak{C}_2' }{\zeta} \int_{\mathbb{R}^3}\mathrm{d}k f(t,k) \varpi_{\mathcal B(k_{n,\Psi_n},\rho_n)}(k)+\mathfrak{C}_2'\zeta^{\alpha'-1},
	\end{aligned}
\end{equation}
for some universal constant $	\mathfrak{C}_2' >0$.

{\it Case (d):}  If $|k|=\min\{|k|,|k_1|,|k_2|,|k_3|\}$,  then since $\omega+\omega_1=\omega_2+\omega_3$, on the support of $\varpi_{\mathcal B(k_{n,\Psi_n},\rho_n)}(k_2)$ we have $\omega\le \omega_3$ and $\omega_1 \ge \omega_2$, yielding $|k_1|\ge |k_2|$. Therefore
\begin{equation}\label{Lemma:Support:E3f}
	\begin{aligned}
		&	\iiiint_{\mathbb{R}^{3\times4}}\mathrm{d}k_1\,\mathrm{d}k_2\,\mathrm{d}k_3\mathrm{d}k \delta(k+k_1-k_2-k_3)\delta(\omega + \omega_1 -\omega_2 - \omega_3)f_2f_3f_1 \chi_{\{|k|=\min\{|k|,|k_1|,|k_2|,|k_3|\}\}}\\
		&\times \left[   \varpi_{\mathcal B(k_{n,\Psi_n},\rho_n)}(k_2)  \ - \ \chi_{\mathcal B(k_{n,\Psi_n},\rho_n)}(k_1) 	 \ - \ \varpi_{\mathcal B(k_{n,\Psi_n},\rho_n)}(k)\chi_{\{|k_1|,|k_3|\le \zeta\}\cup\{|k_1|,|k_2|\le  \zeta\}}\right]\\
		\lesssim\	&		\mathfrak{C}_2'  \int_{\mathbb{R}^3}\mathrm{d}k f(t,k) \varpi_{\mathcal B(k_{n,\Psi_n},\rho_n)}(k),
	\end{aligned}
\end{equation}
for some universal constant $	\mathfrak{C}_2' >0$. On the support of $\varpi_{\mathcal B(k_{n,\Psi_n},\rho_n)}(k_3)$, we have $|k_1|\ge|k_3|$, and
\begin{equation}\label{Lemma:Support:E3f:1}
	\begin{aligned}
		&	\iiiint_{\mathbb{R}^{3\times4}}\mathrm{d}k_1\,\mathrm{d}k_2\,\mathrm{d}k_3\mathrm{d}k \delta(k+k_1-k_2-k_3)\delta(\omega + \omega_1 -\omega_2 - \omega_3)f_2f_3f_1 \chi_{\{|k|=\min\{|k|,|k_1|,|k_2|,|k_3|\}\}}\\
		&\times \left[   \varpi_{\mathcal B(k_{n,\Psi_n},\rho_n)}(k_3)  \right]\
		\lesssim\		\mathfrak{C}_2'  \int_{\mathbb{R}^3}\mathrm{d}k f(t,k) \varpi_{\mathcal B(k_{n,\Psi_n},\rho_n)}(k).
	\end{aligned}
\end{equation}

Combining all the cases, we bound
\begin{equation}\label{Lemma:Support:E3}
	\begin{aligned}
		&		 		 	\iiiint_{\mathbb{R}^{3\times4}}\mathrm{d}k_1\,\mathrm{d}k_2\,\mathrm{d}k_3\mathrm{d}k \delta(k+k_1-k_2-k_3)\delta(\omega + \omega_1 -\omega_2 - \omega_3)f_2f_3f_1 \\
		&\times \Big[   \varpi_{\mathcal B(k_{n,\Psi_n},\rho_n)}(k_2) \ + \ \varpi_{\mathcal B(k_{n,\Psi_n},\rho_n)}(k_3)\\
		&  \ - \ \varpi_{\mathcal B(k_{n,\Psi_n},\rho_n)}(k_1)  \ - \ 
		\varpi_{\mathcal B(k_{n,\Psi_n},\rho_n)}(k)\chi_{\{|k_1|,|k_3|\le \zeta\}\cup\{|k_1|,|k_2|\le \zeta\}}	\Big]\\
		\le \	&	\frac{\mathfrak{C}_2}{\zeta}  \int_{\mathbb{R}^3}\mathrm{d}k f(t,k) \varpi_{\mathcal B(k_{n,\Psi_n},\rho_n)}(k)+\mathfrak{C}_2\zeta^{\alpha'-1},
	\end{aligned}
\end{equation}
for some universal constant $\mathfrak{C}_2>0$ independent of $\zeta$.

{\it Step 2: The second a priori estimate.} 

Next, we will bound
\begin{equation}\label{Lemma:Support:E3:1}
	\begin{aligned}
		&  \iiiint_{\mathbb{R}^{3\times4}}\mathrm{d}k_1\,\mathrm{d}k_2\,\mathrm{d}k_3\mathrm{d}k \delta(k+k_1-k_2-k_3)\delta(\omega + \omega_1 -\omega_2 - \omega_3)f_2f_3f_1 \\
		&\times  \varpi_{\mathcal B(k_{n,\Psi_n},\rho_n)}(k)[1-\chi_{\{|k_1|,|k_3|\le \zeta\}\cup\{|k_1|,|k_2|\le \zeta\}}]\\
		\ge\		& \iiiint_{\mathbb{R}^{3\times4}}\mathrm{d}k_1\,\mathrm{d}k_2\,\mathrm{d}k_3\mathrm{d}k \delta(k+k_1-k_2-k_3)\delta(\omega + \omega_1 -\omega_2 - \omega_3)f_2f_3f_1 \\
		&\times  \left[  \varpi_{\mathcal B(k_{n-1,\Psi_n+2},\rho_{n-1})}(k_1)   \varpi_{\mathcal B(k_{n-1,\Psi_n+1},\rho_{n-1})}(k_2)  \varpi_{\mathcal B(k_{n-1,\Psi_n},\rho_{n-1})}(k_3)
		\right].
	\end{aligned}
\end{equation}
We 
set $x=k_1$ and we define the resonant manifold $\mathscr{S}_{k_2,k_3}$  to be the set of all $x$ that satisfies
\begin{equation}\label{proposition:T2L2:E3}
	\mathscr{G}(x) \ := \ \omega(k_2+k_3-x) + \omega (x) - \omega(k_2)-\omega(k_3)=0.\end{equation}
We can then have a new presentation of the right hand side of \eqref{Lemma:Support:E3:1}
\begin{equation}\begin{aligned}\label{Lemma:Support:E3:2}
		& 
		\iint_{\mathbb{R}^{3\times2}}\mathrm{d}k_3\,\mathrm{d}k_2f_2f_3 \varpi_{\mathcal B(k_{n-1,\Psi_n+1},\rho_{n-1})}(k_2)  \varpi_{\mathcal B(k_{n-1,\Psi_n},\rho_{n-1})}(k_3)\\
		&\ \times\left(\int_{\mathscr{S}_{k_2,k_3}}\mathrm d\mu(x)\frac{\varpi_{\mathcal B(k_{n-1,\Psi_n},\rho_{n-1})}(x)f(x)  }{|\nabla_x \mathscr{G}(x)|} \right).	\end{aligned}
\end{equation}
where $\mu$ is the surface measure on $\mathscr{S}_{k_2,k_3}$.

Setting $k_2+k_3=\gamma$, we compute the gradient of $\mathscr{G}$ 
$$\nabla_x \mathscr{G} = \frac{x - \gamma}{|x-\gamma|} \omega'(|\gamma -x|) + \frac{x}{|x|} \omega'(|x|).$$
We denote $\mathscr V$ to be an arbitrary vector orthogonal to $\gamma$: $\mathscr V \cdot \gamma =0$. We now set  $x=\beta \gamma +\mathscr V, \beta\in\mathbb{R}$. We compute the derivative of $ \mathscr{G}$ in the direction of $\mathscr V$
\begin{equation*} \mathscr V \cdot \nabla_x \mathscr{G} = |\mathscr V|^2 \Big[ \frac{\omega'(|\gamma -x|)}{|\gamma-x|} + \frac{\omega'(|x|) }{|x|}\Big]   > 0.\end{equation*}

The above computation implies that the intersection between the manifold $\mathscr{S}_{k_2,k_2}$ and the  hyperplane $$\mathscr{P}_\beta= \Big\{ \beta\gamma + \mathscr V, \gamma\cdot \mathscr V = 0\Big\}$$ 
can be either empty or the circle centered at $\beta\gamma$. This circle has a radius $\mathscr R_\beta$, with  $\beta \in \mathbb{R}_+$.

The manifold $\mathscr{S}_{k_2,k_3}$ can now be parametrized as  follows. We denote $\gamma^\perp$ to be an arbitrary   vector orthogonal to  $\gamma$. We set  $V_\alpha$ to be a  unit vector in $ \mathscr{P}_0=\{ \gamma\cdot \mathscr V = 0 \}$ and we suppose that the angle between $\gamma^\perp$ and $V_\alpha$ is $\alpha$. The manifold  $\mathscr{S}_{k_2,k_3}$ is now
\begin{equation}
	\Big\{ x_\beta = \beta\gamma + \mathscr R_\beta V_\alpha ~:~ \alpha \in [0,2\pi], ~\beta \in \mathfrak A_{k_2,k_3} \Big\},
\end{equation}
where $\mathfrak A_{k_2,k_3}$ is the set of $\beta$ such that there exists a solution to $\mathscr{G}(x) = 0$. 

The function $\mathscr G$ can now be regarded as  $\mathscr G = \mathscr G( \mathscr R,\beta)$ and $\partial_{\mathscr R} \mathscr G > 0$ for $\mathscr R>0$. By the implicit function theorem, we deduce that the  set of $\mathscr{G}$ can be parameterized as
$$
\{ (\beta, \mathscr R = \mathscr R_\beta), \, \beta \in \mathfrak A_{k_2,k_3}\},
$$
in which $\beta \mapsto \mathscr R_\beta$ is a smooth function on $\mathfrak A_{k_2,k_3}$ vanishing on the boundary. 

Since $\mathscr{G}(x_\beta) =0$ for all $\beta$, we can keep $\alpha$ fixed and compute  
\begin{equation}
	\begin{aligned}
		0 &= \partial_\beta x_\beta \cdot \nabla_x \mathscr{G} =\partial_\beta x_\beta \cdot \left(\frac{x_\beta -\gamma}{|x_\beta -\gamma|}\omega'(|x_\beta -\gamma|)+\frac{x_\beta}{|x_\beta|}\omega'(|x_\beta|)\right)\\
		&=\partial_\beta x_\beta \cdot \left(\frac{x_\beta}{|x_\beta -\gamma|}\omega'(|x_\beta -\gamma|)+\frac{x_\beta}{|x_\beta|}\omega'(|x_\beta|)\right) - \partial_\beta x_\beta \cdot \frac{\gamma}{|x_\beta -\gamma|}\omega'(|x_\beta -\gamma|)\\
		& = \frac12 \partial_\alpha |x_\beta|^2  \Big[ \frac{\omega'(|\gamma -x_\beta|)}{|\gamma-x_\beta|} + \frac{\omega'(|x_\beta|) }{|x_\beta|}\Big]  - |\gamma|^2 \frac{\omega'(|\gamma -x_\beta|)}{|\gamma-x_\beta|},
\end{aligned}\end{equation}
yielding,
\begin{equation}
	\partial_\beta|x_\beta|^2  = 2
	\frac{ \frac{\omega'(|\gamma -x_\beta|)}{|\gamma-x_\beta|}|\gamma|^2 }{ \frac{\omega'(|\gamma -x_\beta|)}{|\gamma-x_\beta|} + \frac{\omega'(|x_\beta|) }{|x_\beta|}}>0,
\end{equation}
which means the change of variables $\beta \to |x_\beta|$ is acceptable.

Since the vector $\partial_\alpha V_\alpha$ is orthogonal to the two vectors $\gamma$ and $ V_\alpha$,  the surface area can be computed
\begin{equation}\begin{aligned}
		\mathrm	d \mu (x) &= |\partial_\beta x \times \partial_\alpha x | \mathrm d\alpha \mathrm d\beta  = \Big |(\gamma+\partial_\beta \mathscr R_\beta V_\alpha) \times   \mathscr R_\beta \partial_\alpha V_\alpha \Big |\mathrm d\alpha \mathrm d\beta
		\\
		&=
		\sqrt{|\gamma|^2  \mathscr R_\beta ^2+ \frac14| \partial_\beta (  \mathscr R_\beta^2)|^2}\mathrm d\alpha\mathrm d\beta .
\end{aligned}\end{equation}
Since $$|x|^2 = \beta^2 |\gamma|^2 + 
\mathscr R_\beta^2,$$
we deduce
\begin{equation}
	\begin{aligned}
		\partial_\beta  \mathscr R_\beta^2  &
		&= 
		2|\gamma|^2 \frac{\beta   \frac{\omega'(|x_\beta|) }{|x_\beta|} +(\beta-1)  \frac{\omega'(|\gamma -z_\beta|)}{|\gamma-x_\beta|}}{\frac{\omega'(|\rho -x_\beta|)}{|\gamma-x_\beta|} + \frac{\omega'(x_\beta) }{|x_\beta|}}.
\end{aligned}\end{equation}
Next, we compute $|\nabla_x\mathscr{G}|$  
\begin{equation*}
	\begin{aligned}
		|\nabla_x \mathscr{G}|^2 \ =  &\left|\frac{x_\beta}{|x_\beta|}\omega'(|x_\beta|)+\frac{x_\beta -\gamma}{|x_\beta -\gamma|}\omega'(|x_\beta -\gamma|)\right|^2\\
		= & \left|\frac{\beta \gamma+V}{|x_\beta|}\omega'(|x_\beta|)+\frac{(\beta-1)\gamma+q}{|x_\beta -\gamma|}\omega'(|x_\beta -\gamma|)\right|^2\\
		\ =  &\  |\gamma|^2\left[\beta   \frac{\omega'(|x_\beta|) }{|z_\beta|} +(\beta-1)  \frac{\omega'(|\gamma -x_\beta|)}{|\gamma-x_\beta|}\right]^2 + \mathscr R_\beta^2 \left[\frac{\omega'(|\gamma -x_\beta|)}{|\gamma-x_\beta|} + \frac{\omega'(|x_\beta|) }{|x_\beta|}\right]^2,
	\end{aligned}
\end{equation*}
which implies
\begin{equation} 
	\begin{aligned}
		|\nabla_x \mathscr{G}|^2  \ =  &\  \frac{\left|\partial_\beta \mathscr R_\beta^2\right|^2}{4|\gamma|^2}\left[\frac{\omega'(|\gamma -x_\beta|)}{|\gamma-x_\beta|} + \frac{\omega'(|x_\beta|) }{|x_\beta|}\right]^2 +\mathscr R_\beta^2 \left[\frac{\omega'(|\gamma -x_\beta|)}{|\gamma-x_\beta|} + \frac{\omega'(|x_\beta|) }{|x_\beta|}\right]^2.
	\end{aligned}
\end{equation}
Therefore, we deduce
\begin{equation}\label{proposition:T2L2:E9a}\begin{aligned}
		\frac{ \mathrm d \mu (x)}{ |\nabla_x \mathscr{G}|} &= \frac{|\gamma|}{\frac{\omega'(|\gamma -x_\beta|)}{|\gamma-x_\beta|} + \frac{\omega'(|x_\beta|) }{|x_\beta|}}\, \mathrm  d\alpha\,\mathrm d\beta.
\end{aligned}\end{equation}
We then introduce $|x | = \sqrt{\beta^2 |\gamma|^2 +\mathscr R_\beta^2}$  and conclude
$$
\frac{\mathrm d \mu (x)}{ |\nabla_x \mathscr{G}|} = \frac{|\gamma-x|}{\omega'(|\gamma -x|) |\gamma|} |x|\,\mathrm d|x| \,\mathrm d\alpha,
$$
which, in combination with \eqref{Lemma:Support:E3:1}, \eqref{Lemma:Support:E3:2}, yields
\begin{equation}\label{Lemma:Support:E3:3}
	\begin{aligned}
		&  \iiiint_{\mathbb{R}^{3\times4}}\mathrm{d}k_1\,\mathrm{d}k_2\,\mathrm{d}k_3\mathrm{d}k \delta(k+k_1-k_2-k_3)\delta(\omega + \omega_1 -\omega_2 - \omega_3)f_2f_3f_1   \varpi_{\mathcal B(k_{n,\Psi_n},\rho_n)}(k)\\
		\ge\		&	 
		\iint_{\mathbb{R}^{3\times2}}\mathrm{d}k_3\,\mathrm{d}k_2f_2f_3\varpi_{\mathcal B(k_{n-1,\Psi_n+1},\rho_{n-1})}(k_2)  \varpi_{\mathcal B(k_{n-1,\Psi_n },\rho_{n-1})}(k_3)\\
		&\ \times\left( \int_0^{2\pi}\mathrm d\alpha\int_{A_{k_2,k_3}}^{B_{k_2,k_3}}\mathrm d|x|\frac{\varpi_{\mathcal B(k_{n-1,\Psi_n+2},\rho_{n-1})}(x)f(|x|) |k_2+k_3 -x| |x|}{\omega'(|k_2+k_3 -x|) |k_2+k_3| } \right),
	\end{aligned}
\end{equation}
where ${A_{k_2,k_3}},{B_{k_2,k_3}}$ are the boundaries for the integration in $\mathrm d|x|$. By our assumption on $\omega$, we can then bound
\begin{equation}\label{Lemma:Support:E3:5}
	\begin{aligned}
		&  \iiiint_{\mathbb{R}^{3\times4}}\mathrm{d}k_1\,\mathrm{d}k_2\,\mathrm{d}k_3\mathrm{d}k \delta(k+k_1-k_2-k_3)\delta(\omega + \omega_1 -\omega_2 - \omega_3)f_2f_3f_1   \varpi_{\mathcal B(k_{n,\Psi_n},\rho_n)}(k)\\
		\gtrsim \		&	 
		\iint_{\mathbb{R}^{3\times2}}\mathrm{d}k_3\,\mathrm{d}k_2f_2f_3 \varpi_{\mathcal B(k_{n-1,\Psi_n+1},\rho_{n-1})}(k_2)  \varpi_{\mathcal B(k_{n-1,\Psi_n },\rho_{n-1})}(k_3)\\
		&\ \times\left( \int_0^{2\pi}\mathrm d\alpha\int_{A_{k_2,k_3}}^{\mathcal B_{k_2,k_3}}\mathrm d|x|\frac{\varpi_{\mathcal B(k_{n-1,\Psi_n+2},\rho_{n-1})}(x)f(|x|) |x| }{ |k_2+k_3| } \right).
	\end{aligned}
\end{equation}
According to our construction, 
\begin{equation}
	\label{Lemma:Support:E3:6}
	\omega(k_{n-1,\Psi_n+1})\ +\ \omega(k_{n-1,\Psi_n}) \ = \ \omega(k_{n-1,\Psi_n+1}+k_{n-1,\Psi_n}-k_{n-1,\Psi_n+2}) \ + \ \omega(k_{n-1,\Psi_n+2}),
\end{equation}
we construct a new function
\begin{equation}
	\label{Lemma:Support:E3:6}\begin{aligned}
		F(x_1,x_2)\  :=\ &	\omega(k_{n-1,\Psi_n+1}+x_2)\ +\ \omega(k_{n-1,\Psi_n}-x_2)\\
		& \ - \ \omega(k_{n-1,\Psi_n+2}+x_1) \ - \ \omega(k_{n-1,\Psi_n+1}+k_{n-1,\Psi_n}-k_{n-1,\Psi_n+2}-x_1),\end{aligned}
\end{equation}
then $F(0,0)=0.$ Suppose that $x_1=(x_1^1,x_1^2,x_1^3)$ and $x_2=(x_2^1,x_2^2,x_2^3)$
we compute the partial derivative in $x_1^i$, with $i=1,2,3$, 
\begin{equation}
	\label{Lemma:Support:E3:7}\begin{aligned}
		&	\partial_{x_1^i}	F(x_1,x_2)\  =\ \ - \ \frac{\omega'(|k_{n-1,\Psi_n+2}+x_1|)}{|k_{n-1,\Psi_n+2}+x_1|}(k_{n-1,\Psi_n+2}^i+x_1^i)\\
		& \ + \ \frac{\omega'(|k_{n-1,\Psi_n+1}+k_{n-1,\Psi_n}-k_{n-1,\Psi_n+2}-x_1|)}{|k_{n-1,\Psi_n+1}+k_{n-1,\Psi_n}-k_{n-1,\Psi_n+2}-x_1|}(k_{n-1,\Psi_n+1}^i+k^i_{n-1,\Psi_n}-k_{n-1,\Psi_n+2}^i-x_1^i),\end{aligned}
\end{equation}
in which $k_{n-1,\Psi_n+1}=(k_{n-1,\Psi_n+1}^1,$ $k_{n-1,\Psi_n+1}^2,k_{n-1,\Psi_n+1}^3)$, $k_{n-1,\Psi_n+2}=(k_{n-1,\Psi_n+2}^1,k_{n-1,\Psi_n+2}^2,$ $k_{n-1,\Psi_n+2}^3)$, $k_{n-1,\Psi_n}=(k_{n-1,\Psi_n}^1,k_{n-1,\Psi_n}^2,$ $k_{n-1,\Psi_n}^3)$.

Replacing $(x_1,x_2)=(0,0)$ in the above computation gives
\begin{equation}
	\label{Lemma:Support:E3:9}\begin{aligned}
		&	\partial_{x_1^i}	F(0,0)\  =\ \ - \ \frac{\omega'(|k_{n-1,\Psi_n+2}|)}{|k_{n-1,\Psi_n+2}|}(k_{n-1,\Psi_n+2}^i)\\
		& \ + \ \frac{\omega'(|k_{n-1,\Psi_n+1}+k_{n-1,\Psi_n}-k_{n-1,\Psi_n+2}|)}{|k_{n-1,\Psi_n+1}+k_{n-1,\Psi_n}-k_{n-1,\Psi_n+2}|}(k_{n-1,\Psi_n+1}^i+k^i_{n-1,\Psi_n}-k_{n-1,\Psi_n+2}^i).\end{aligned}
\end{equation}
We will show that  $(	\partial_{x_1^1}	F(0,0),\partial_{x_1^2}	F(0,0),\partial_{x_1^3}	F(0,0))
$ is not zero. Suppose the contrary, we have
\begin{equation}
	\label{Lemma:Support:E3:9a}\begin{aligned}
		&	  \frac{\omega'(|k_{n-1,\Psi_n+2}|)}{|k_{n-1,\Psi_n+2}|}(k_{n-1,\Psi_n+2} )\\
\ = \		&  \frac{\omega'(|k_{n-1,\Psi_n+1}+k_{n-1,\Psi_n}-k_{n-1,\Psi_n+2}|)}{|k_{n-1,\Psi_n+1}+k_{n-1,\Psi_n}-k_{n-1,\Psi_n+2}|}(k_{n-1,\Psi_n+1}+k_{n-1,\Psi_n}-k_{n-1,\Psi_n+2} ),\end{aligned}
\end{equation}
which implies $\omega'(|k_{n-1,\Psi_n+2}|)=\omega'(|k_{n-1,\Psi_n+1}+k_{n-1,\Psi_n}-k_{n-1,\Psi_n+2}|)$, yielding $|k_{n-1,\Psi_n+2}|=|k_{n-1,\Psi_n+1}+k_{n-1,\Psi_n}-k_{n-1,\Psi_n+2}|$. Hence, from \eqref{Lemma:Support:E3:9a}, we obtain $k_{n-1,\Psi_n+2}=k_{n-1,\Psi_n+1}+k_{n-1,\Psi_n}-k_{n-1,\Psi_n+2}$, yielding
$$\omega(|k_{n-1,\Psi_n+1}|) \ - \ \omega(|k_{n-1,\Psi_n+2}|) \ = \ \omega(|k_{n-1,\Psi_n+2}|) \ - \ \omega(|2k_{n-1,\Psi_n+2}-k_{n-1,\Psi_n+1}|),$$
which leads to
$$\int_{|k_{n-1,\Psi_n+2}|}^{|k_{n-1,\Psi_n+1}|}\mathrm{d}|\xi_1|\omega'(|\xi_1|)   \ = \ \int_{|2k_{n-1,\Psi_n+2}-k_{n-1,\Psi_n+1}|}^{|k_{n-1,\Psi_n+2}|}\mathrm{d}|\xi_2|\omega'(|\xi_2|).$$
Since $ |k_{n-1,\Psi_n+1}|  \ - \  |k_{n-1,\Psi_n+2}|  \ \ge  \  |k_{n-1,\Psi_n+2}|  \ - \  |2k_{n-1,\Psi_n+2}-k_{n-1,\Psi_n+1}| ,$ 
and by the convexity of $\omega(|k|)$, we deduce $k_{n-1,\Psi_n+1}=k_{n-1,\Psi_n}=k_{n-1,\Psi_n+2}$. That contradicts our construction. 

Since the vector $(	\partial_{x_1^1}	F(0,0),\partial_{x_1^2}	F(0,0),\partial_{x_1^3}	F(0,0))
$ is not zero, there exists $j\in\{1,2,3\}$ such that $	\partial_{x_1^j}	F(0,0)\ne 0$. By the implicit function theorem, there exits an interval $(-\delta,\delta)$ and a function $g:(-\delta,\delta)\to \mathbb{R}$ such that for any $\epsilon\in (-\delta,\delta)$, if we define the two new vectors $X_1=(X_1^1,X_1^2,X_1^3)$, $X_2=(X_2^1,X_2^2,X_2^3)$ as follows $X_1^j(\epsilon)=\epsilon$ and $X_1^i(\epsilon)=0$ for $i\in \{1,2,3\}\backslash\{j\}$, $X_2^j(\epsilon)=g(\epsilon)$, $X_2^i(\epsilon)=0$ for $i\in \{1,2,3\}\backslash\{j\}$, then $F(X_1,X_2)=0.$ Since we can adjust $\rho$ such that $\rho_{n-1}<\delta$, we bound \eqref{Lemma:Support:E3:5} as
\begin{equation}\label{Lemma:Support:E3:10}
	\begin{aligned}
		&  \iiiint_{\mathbb{R}^{3\times4}}\mathrm{d}k_1\,\mathrm{d}k_2\,\mathrm{d}k_3\mathrm{d}k \delta(k+k_1-k_2-k_3)\delta(\omega + \omega_1 -\omega_2 - \omega_3)f_2f_3f_1   \varpi_{\mathcal B(k_{n,\Psi_n},\rho_n)}(k)\\
		\gtrsim \		&	 
		\iint_{\mathbb{R}^{3\times2}}\mathrm{d}k_3\,\mathrm{d}k_2f_2f_3 \chi_{\mathcal B(k_{n-1,\Psi_n+1},\rho_{n-1})}(k_2)  \varpi_{\mathcal B(k_{n-1,\Psi_n+2},\rho_{n-1})}(k_3)\\
		&\ \times\left( \int_0^{2\pi}\mathrm d\alpha\int_{|k_{n-1,\Psi_n+2}|-v}^{|k_{n-1,\Psi_n+2}| +v }\mathrm d|x|\frac{f(|x|)  |x|}{ |k_2+k_3| } \right),
	\end{aligned}
\end{equation}
for some constant  $ 0<v<\delta$.
Therefore,
\begin{equation}\label{Lemma:Support:E4}
	\begin{aligned}
		& \iiiint_{\mathbb{R}^{3\times4}}\mathrm{d}k_1\,\mathrm{d}k_2\,\mathrm{d}k_3\mathrm{d}k \delta(k+k_1-k_2-k_3)\delta(\omega + \omega_1 -\omega_2 - \omega_3)f_2f_3(f+f_1) \\
		&\times  \left[ \varpi_{\mathcal B(k_{n-1,\Psi_n},\rho_{n-1})}(k_1)   \varpi_{\mathcal B(k_{n-1,\Psi_n+1},\rho_{n-1})}(k_2)  \varpi_{\mathcal B(k_{n-1,\Psi_n+2},\rho_{n-1})}(k_3)
		\right]\\
		\ \ge \  & \mathfrak{C}_1 \left[\int_{\mathbb{R}^{3}}\mathrm{d}k f(t,k) \varpi_{\mathcal B(k_{n-1,\Psi_n},\rho_{n-1})}(k) \right] \left[\int_{\mathbb{R}^{3}}\mathrm{d}k f(t,k) \varpi_{\mathcal B(k_{n-1,\Psi_n+1},\rho_{n-1})}(k)\right]\\
		&\times \left[\int_{\mathbb{R}^{3}}\mathrm{d}k f(t,k) \varpi_{\mathcal B(k_{n-1,\Psi_n+2},\rho_{n-1})}(k)\right],
	\end{aligned}
\end{equation}
for some universal constant $\mathfrak{C}_1>0$.

	{\it Step 3: An induction argument.}

Let $t_0=\zeta_0\tau$, where  $0<\tau<1$ and $\zeta_0$ is a small parameter that will be specified later. 
We follow the inductive argument in $n$ introduced in  \cite{ToanBinh}. Let us emphasize that the induction only runs for a finite number of steps, where $n$ goes from $0$ to $M$. For $n=0$, using \eqref{Lemma:Support:E3}, we have
\begin{equation}\label{Lemma:Support:E6}
	\begin{aligned}
		 \partial_t\int_{\mathbb{R}^3}\mathrm{d}k f(t,k) \varpi_{\mathcal B(k_{0,i},\rho_0)}(k) 		 \ge \  &  -\frac{\mathfrak{C}_2}{\zeta} \int_{\mathbb{R}^3}\mathrm{d}k f(t,k) \varpi_{\mathcal B(k_{0,i},\rho_0)}(k) - \mathfrak{C}_2\zeta^{\alpha'-1},
	\end{aligned}
\end{equation}
with $i=0,\cdots, M$. 

 There exists $c_0>0$ such that $$\int_{\mathbb{R}^3}\mathrm{d}k f(0,k) \varpi_{\mathcal B(k_{0,i},\rho_0)}(k)>c_0$$ by our choice   $k_{0,i}\in \mathrm{supp}f_0$, with $i=0,\cdots, M$, we deduce  
 $$\int_{\mathbb{R}^3}\mathrm{d}k f(t,k) \varpi_{\mathcal B(k_{0,i},\rho_0)}(k){e^{\frac{\mathfrak{C}_2}{\zeta}{t}}}\ge \int_{\mathbb{R}^3}\mathrm{d}k f(0,k) \varpi_{\mathcal B(k_{0,i},\rho_0)}(k)-\zeta^{\alpha'}{e^{\frac{\mathfrak{C}_2}{\zeta}{t}}}.$$
 When $\zeta=\zeta_0$, which is sufficiently small, we bound

 $$\int_{\mathbb{R}^3}\mathrm{d}k f(t,k) \varpi_{\mathcal B(k_{0,i},\rho_0)}(k)>c_0'\zeta_0:=c_0e^{-m_0{\mathfrak{C}_2}{\tau}}\zeta_0$$ for all $i=0,\cdots, M$ and $t_0\le t\le m_0t_0$, $m_0>2$ is a positive number.

Suppose  $$\int_{\mathbb{R}^3}\mathrm{d}k f(t,k) \varpi_{\mathcal B(k_{n-1,i},\rho_{n-1})}(k)>c_{n-1}'\zeta_0^{9^{n-1}}>0$$ for all $i=0,\cdots, M-n+1$, $t_{n-1}\le t\le m_0t_0$  and $c_{n-1}'$ is independent of $\zeta$. We will prove that  there exists $t_n>t_{n-1}$  \begin{equation}\label{Lemma:Support:E6:1}\int_{\mathbb{R}^3}\mathrm{d}k f(t,k) \varpi_{\mathcal B(k_{n,\Psi_{n}},\rho_{n})}(k)>c_{n}''\zeta_0^{9^{n}}>0,\end{equation}
and $c_n''$ is independent of $\zeta$, for
$t_n\le t\le m_0t_0$.
Note that once the above inequality is proved, we deduce $$\int_{\mathbb{R}^3}\mathrm{d}k f(t_n,k) \varpi_{\mathcal B(k_{n,i},\rho_{n})}(k)>c_{n}'\zeta_0^{9^{n}}>0$$ for all $i=0,\cdots, M-n$, $t_n\le t\le m_0 t_0$.

	Putting together \eqref{Lemma:Support:E1}-\eqref{Lemma:Support:E2}-\eqref{Lemma:Support:E4}-\eqref{Lemma:Support:E3}, we find, for $t_{n-1}\le t\le m_0 t_0$
		\begin{equation}\label{Lemma:Support:E5}
			\begin{aligned}
				&\partial_t  \int_{\mathbb{R}^3}\mathrm{d}k f(t,k) \varpi_{\mathcal B(k_{n,\Psi_n},\rho_n)}(k)\\
			\ \ge \  &\mathfrak{C}_1 \left[\int_{\mathbb{R}^{3}}\mathrm{d}k f(t,k) \varpi_{\mathcal B(k_{n-1,\Psi_n},\rho_{n-1})}(k) \right] \left[\int_{\mathbb{R}^{3}}\mathrm{d}k f(t,k) \varpi_{\mathcal B(k_{n-1,\Psi_n+1},\rho_{n-1})}(k)\right]\\
			&\times \left[\int_{\mathbb{R}^{3}}\mathrm{d}k f(t,k) \varpi_{\mathcal B(k_{n-1,\Psi_n+2},\rho_{n-1})}(k)\right] -\frac{\mathfrak{C}_2}{\zeta}\int_{\mathbb{R}^3}\mathrm{d}k f(t,k) \varpi_{\mathcal B(k_{n,\Psi_n},\rho_n)}(k)- \mathfrak{C}_2\zeta^{\alpha'-1}.
			\\
			\ \ge \  &\mathfrak{C}_1 (c'_{n-1})^3\zeta_0^{3.9^{n-1}}-\frac{\mathfrak{C}_2}{\zeta}\int_{\mathbb{R}^3}\mathrm{d}k f(t,k) \varpi_{\mathcal B(k_{n,\Psi_n},\rho_n)}(k)- \mathfrak{C}_2\zeta^{\alpha'-1},
			\end{aligned}
		\end{equation}
	for $0<\zeta<1$.	
		Since in the above equations $\mathfrak{C}_1,\mathfrak{C}_2,c'_{n-1}$ are indepedent of $\zeta$, choosing $$\zeta^{\alpha'-1}=\frac{1}{2}\mathfrak{C}_1 (c'_{n-1})^3\zeta_0^{3.9^{n-1}}/\mathfrak{C}_2,$$ where $\zeta_0$ is sufficiently small, we have $$\mathfrak{C}_1 (c'_{n-1})^3\zeta_0^{3.9^{n-1}}-\mathfrak{C}_2\zeta^{\alpha'-1}=\frac{1}{2}\mathfrak{C}_1 (c'_{n-1})^3\zeta_0^{3.9^{n-1}},$$ we bound
		\begin{equation}\label{Lemma:Support:E5:1}
			\begin{aligned}
				\partial_t  \int_{\mathbb{R}^3}\mathrm{d}k f(t,k) \varpi_{\mathcal B(k_{n,\Psi_n},\rho_n)}(k)\			\ \ge \  &\frac{\mathfrak{C}_1 }{2}(c'_{n-1})^3\zeta_0^{3.9^{n-1}}-\frac{\mathfrak{C}_2}{\zeta}\int_{\mathbb{R}^3}\mathrm{d}k f(t,k) \varpi_{\mathcal B(k_{n,\Psi_n},\rho_n)}(k),
			\end{aligned}
		\end{equation}
		which implies 
				\begin{equation}\label{Lemma:Support:E5:2}
			\begin{aligned}
				 \int_{\mathbb{R}^3}\mathrm{d}k f(t,k) \varpi_{\mathcal B(k_{n,\Psi_n},\rho_n)}(k)			\ \ge \  &\frac{\mathfrak{C}_1 }{2\mathfrak{C}_2}(c'_{n-1})^3\zeta_0^{3.9^{n-1}}\Big[\frac{1}{2\mathfrak{C}_2}\mathfrak{C}_1 (c'_{n-1})^3\zeta_0^{3.9^{n-1}}\Big]^{\frac{1}{\alpha'-1}}[1-e^{-\frac{\mathfrak C_2}{\zeta}(t-t_{n-1})}]\\\ \ge \  &\frac{\mathfrak{C}_1 }{2\mathfrak{C}_2}(c'_{n-1})^3\zeta_0^{3.9^{n-1}}\Big[\frac{1}{2\mathfrak{C}_2}\mathfrak{C}_1 (c'_{n-1})^3\zeta_0^{3.9^{n-1}}\Big]^{\frac{1}{\alpha'-1}}[1-e^{-\frac{\mathfrak C_2}{\zeta}(t_n-t_{n-1})}]\\
				 \ \ge \  &\frac{\mathfrak{C}_1 }{2\mathfrak{C}_2}(c'_{n-1})^3\zeta_0^{3.9^{n-1}}\Big[\frac{1}{2\mathfrak{C}_2}\mathfrak{C}_1 (c'_{n-1})^3\zeta_0^{3.9^{n-1}}\Big]^{\frac{1}{\alpha'-1}}[1-e^{-t_0}]\\
				 \ > \  &\frac{\mathfrak{C}_1 }{2\mathfrak{C}_2}(c'_{n-1})^3\Big[\frac{1}{2\mathfrak{C}_2}\mathfrak{C}_1 (c'_{n-1})^3\Big]^{\frac{1}{\alpha'-1}}\zeta_0^{9^{n}}[1-e^{-t_0}],
			\end{aligned}
		\end{equation}
		for $t_n=t_{n-1}(1+\zeta/\mathfrak{C}_2)$ , $t_{n}\le t\le m_0 t_0$.

	Therefore \eqref{Lemma:Support:E6:1} is proved.
 By induction, there exists $m_0t_0>t_M>0$ such that  
	$$\int_{\mathbb{R}^3}\mathrm{d}k f(t,k) \varpi_{\mathcal B(k_{0},\rho)}(k)>c_{M}'\zeta_0^{9^{M}},$$ 
for $t_{M}\le t\le m_0 t_0$. Note that since the number of inductive steps is finite, $\zeta_0$ can be chosen such that $t_0<t_1<\cdots<t_{M}< m_0 t_0$.
	
 As a result, $k_0$ belongs to the support of $f(T_0,k)$, for $T_0=t_M$. By Lemma \ref{Lemma:CollisionRegionSphere}, we conclude our proof.
\end{proof}

\section{Global well-posedness}

\begin{proposition}
	\label{Propo:GlobalWellPosed} Let $f_0(k)=f_0(|k|)\ge 0$ be an initial condition that satisfies
	\begin{equation}
\label{Propo:GlobalWellPosed:1} \int_{\mathbb{R}^3}\mathrm{d}k f_0(k) \ = \ \mathscr{M}.
	\end{equation}
There exists a global mild radial solution $f(t,k)$ of \eqref{4wave} in the sense of \eqref{4wavemild} such that
\begin{equation}
	\label{Propo:GlobalWellPosed:2} \int_{\mathbb{R}^3}\mathrm{d}k f(t,k) \ = \ \mathscr{M},
\end{equation}
for all $t\ge 0$. 
\end{proposition}

\begin{proof} We divide the proof into several steps.

	{\it Step 1: Change of variables.}

Performing the change of variables $|k|\to\omega$, $|k_1|\to\omega_1$, $|k_2|\to\omega_2$, $|k_3|\to\omega_3$, we convert the weak form of \eqref{Lemma:TestFunction:E4} into
\begin{equation}\label{Lemma:TestFunction:E5}
	\begin{aligned}
		\int_{\mathbb{R}_+}\mathrm{d}\omega \mathcal Q \left[ f\right]  \phi(\omega)|k| \mho
		\ =\ & C_Q\iiiint_{\mathbb{R}_+^{4}}\mathrm{d}\omega_1\,\mathrm{d}\omega_2\,\mathrm{d}\omega_3\mathrm{d}\omega \delta(\omega + \omega_1 -\omega_2 - \omega_3)[f_2f_3(f_1+f)-ff_1(f_2+f_3)]\\
		&\times {\mho\mho_1\mho_2\mho_3\min\{|k_1|,|k_2|,|k_3|,|k|\}}\phi,
	\end{aligned}
\end{equation}
for any test function $\phi(\omega)\in C(\mathbb{R}_+)$, in which $$ \mho=\mho(\omega)=\frac{|k|}{\omega'(|k|)} ,  \ \  \mho_1=\mho(\omega_1)=\frac{|k_1|}{\omega'(|k_1|)},   \ \ \mho_2=\mho(\omega_2)=\frac{|k_2|}{\omega'(|k_2|)}, \ \ \mho_3=\mho(\omega_3)=\frac{|k_3|}{\omega'(|k_3|)}.$$ 
Note that now $|k|$ is a function of $\omega$.
We set $$\Xi(\omega,\omega_1,\omega_2,\omega_3)={\mho\mho_1\mho_2\mho_3\min\{|k_1|,|k_2|,|k_3|,|k|\}}$$ deduce from \eqref{Lemma:TestFunction:E5} that 
\begin{equation}\label{Lemma:TestFunction:E6}
	\begin{aligned}
		\int_{\mathbb{R}_+}\mathrm{d}\omega \mathcal Q \left[ f\right]  \phi(\omega)|k| \mho
		\ =\ & C_Q\iiiint_{\mathbb{R}_+^{4}}\mathrm{d}\omega_1\,\mathrm{d}\omega_2\,\mathrm{d}\omega_3\mathrm{d}\omega \delta(\omega + \omega_1 -\omega_2 - \omega_3)f_1f_2f\Xi(\omega,\omega_1,\omega_2,\omega_3)\\
		&\times [-\phi(\omega)-\phi(\omega_1)+\phi(\omega_2)+\phi(\omega_3)]\\
		\ =\ & C_Q\iiint_{\mathbb{R}_+^{3}}\mathrm{d}\omega_1\,\mathrm{d}\omega_2\,\mathrm{d}\omega f_1f_2f\Xi(\omega,\omega_1,\omega_2,\omega+\omega_1-\omega_2)\mathbf{1}_{\omega+\omega_1\ge \omega_2}\\
		&\times [-\phi(\omega)-\phi(\omega_1)+\phi(\omega_2)+\phi(\omega+\omega_1-\omega_2)].
	\end{aligned}
\end{equation}

Setting $g(\omega)=\mho {f(\omega)}{|k|}$, $g_1=g(\omega_1)$, $g_2=g(\omega_2)$, we get

\begin{equation}\label{Lemma:GlobalExist:E1}
	\begin{aligned}
		\int_{\mathbb{R}_+}\mathrm{d}\omega \mathcal Q \left[ f\right]  \phi(\omega)|k| \mho
		\ =\ & C_Q\iiint_{\mathbb{R}_+^{3}}\mathrm{d}\omega_1\,\mathrm{d}\omega_2\,\mathrm{d}\omega g_1g_2g\mho_3\frac{\min\{|k_1|,|k_2|,|k_3|,|k|\}}{|k||k_1||k_2|}\\
		&\times [-\phi(\omega)-\phi(\omega_1)+\phi(\omega_2)+\phi(\omega+\omega_1-\omega_2)]\mathbf{1}_{\omega+\omega_1\ge \omega_2}.
	\end{aligned}
\end{equation}

{\it Step 2: Boundedness of the kernel.}
We define the cut off kernel, for $n\in(0,\infty]$
\begin{equation}\label{Lemma:GlobalExist:E2}
	\begin{aligned}
		\mathcal K_n
		\ =\ & \mho_3 \frac{\min\{|k_1|,|k_2|,|k_3|,|k|,n\}\chi_{\Big[\frac1n,n\Big)}(|k|)\chi_{\Big[\frac1n,n\Big)}(|k_1|)\chi_{\Big[\frac1n,n\Big)}(|k_2|)}{|k||k_1||k_2|\mathbf{1}_{\omega+\omega_1\ge \omega_2}},
	\end{aligned}
\end{equation}
and consider the approximated equation 
\begin{equation}\label{Lemma:GlobalExist:E3}
	\begin{aligned}
		\int_{\mathbb{R}_+}\mathrm{d}\omega \partial_t g^n  \phi 
		\ =\ & C_Q\iiint_{\mathbb{R}_+^{3}}\mathrm{d}\omega_1\,\mathrm{d}\omega_2\,\mathrm{d}\omega g_1^ng_2^ng^n\mathcal K_n  [-\phi(\omega)-\phi(\omega_1)+\phi(\omega_2)+\phi(\omega+\omega_1-\omega_2)],
	\end{aligned}
\end{equation}
with the initial condition $g^n(0,\omega)=g(0,\omega)=\mho {f_0(\omega)}{|k|}$. It is straightforward that
\begin{equation}\label{Lemma:GlobalExist:E3a}
	\begin{aligned}
		\int_{\mathbb{R}_+}\mathrm{d}\omega \partial_t g^n   
		\ =\ & 0,
	\end{aligned}
\end{equation}
yielding 
\begin{equation}\label{Lemma:GlobalExist:E3b}
	\begin{aligned}
		\int_{\mathbb{R}_+}\mathrm{d}\omega   g^n(t)  
		\ =\ & 	\int_{\mathbb{R}_+}\mathrm{d}\omega   g^n(0) \ =\  \int_{\mathbb{R}_+}\mathrm{d}\omega   \mho {f_0(\omega)}{|k|} \ = \ \mathscr M.
	\end{aligned}
\end{equation}
Next, we will show that for $\phi\in C_c^2(\mathbb{R}_+)$, the quantity 
$$\mathfrak{K}_n=\mathcal K_n  [-\phi(\omega)-\phi(\omega_1)+\phi(\omega_2)+\phi(\omega+\omega_1-\omega_2)]$$
 is continuous and uniformly bounded for all $n\in(0,\infty]$. It should be clear that we only need to derive a uniform bound for $n=\infty$. To this end, we only need to prove that $\mathfrak{K}_\infty$ is uniformly bounded near the singularities $|k|=0$, $|k_1|=0$, $|k_2|=0$.

Let us first consider the case when only one of the quantities $|k|,|k_1|,|k_2|$ is close to $0$ while the other two are bounded from below by a constant  $c_o>0$. We suppose without loss of generality that $|k|$ is close to $0$ and $|k_1|,|k_2|\ge c_o$, then $$\mathcal K_\infty\ \le\ \frac{1}{|k_1| |k_2|}\mho_3 \mathbf{1}_{\omega_1\ge \omega_2}\ \lesssim\  \frac{1}{c_o^2}  \mathbf{1}_{\omega_1\ge \omega_2},$$  
which is bounded. 

Next, we consider the case when two of the quantities $|k|,|k_1|,|k_2|$ are close to $0$ while the other one is  bounded from below by a constant  $c_o>0$. Since our domain of integration is restricted to $\omega+\omega_1\ge \omega_2$, the case when both $|k|,|k_1|$ are close to  $0$ while $|k_2|$ is bounded from below by $c_o>0$ does not happen. Since $k,k_1$ are symmetric, we only consider the case when $|k_2|,|k|$ are close to $0$ while $|k_1|$ is bounded from below by $c_o$. We write
\begin{equation}\begin{aligned}
		\mathfrak{K}_\infty\ =\ & \mathcal K_\infty  [-\phi(\omega)-\phi(\omega_1)+\phi(\omega_2)+\phi(\omega+\omega_1-\omega_2)]\\
		\ =\ &\ \mathcal K_\infty  \Big[-\int_0^{\omega-\omega_2}\mathrm{d}\xi\phi'(\omega_2+\xi) \ +\ \int_0^{\omega-\omega_2}\mathrm{d}\xi \phi'( \omega_1+\xi)\Big]
		\\
		\ =\ &\ \mathcal K_\infty  \Big[-\int_0^{\omega-\omega_2}\mathrm{d}\xi[\phi'(\omega_2+\xi)-\phi'(0)] \ -\ \int_0^{\omega-\omega_2}\mathrm{d}\xi  \phi'(0)\\
		&  \ +\ \int_0^{\omega-\omega_2}\mathrm{d}\xi [\phi'( \omega_1+\xi)-\phi'( \omega_1)] \ + \ \int_0^{\omega-\omega_2}\mathrm{d}\xi \phi'( \omega_1) \Big]\\
		\ =\ & \mathcal K_\infty  \Big[-\int_0^{\omega-\omega_2}\mathrm{d}\xi\int_0^{\omega_2+\xi}\mathrm{d}\xi'\phi''(\xi') \ -\  \phi'(0)(\omega-\omega_2)\\
		&  \ +\ \int_0^{\omega-\omega_2}\mathrm{d}\xi\int_0^{\xi}\mathrm{d}\xi' \phi''( \omega_1+\xi') \ + \   \phi'( \omega_1)(\omega-\omega_2) \Big],
	\end{aligned}
\end{equation}
which can be bounded as

\begin{equation}\label{Lemma:GlobalExist:E4}
	\begin{aligned}
		|\mathfrak{K}_\infty|\ \leq\ &
		\mathcal K_\infty  \left(  \left\vert
		\phi^{\prime}\left(  \omega_{1}\right)  \left(  \omega-\omega
		_{2}\right)  -\phi^{\prime}\left(  0\right)  \left(  \omega-\omega
		_{2}\right)  \right\vert +C\left[  \left(  \omega\right)  ^{2}+\left(
		\omega_{2}\right)  ^{2}\right]  \right)\\
		\ \lesssim\ &
		\frac{\min\{|k_1|,|k_2|,|k_3|,|k|\}}{|k||k_1||k_2|}\mathbf{1}_{\omega+\omega_1\ge \omega_2} \left(  \left\vert
		\phi^{\prime}\left(  \omega_{1}\right)  \left(  \omega-\omega
		_{2}\right)\right\vert  +\left\vert\phi^{\prime}\left(  0\right)  \left(  \omega-\omega
		_{2}\right)  \right\vert +C\left[  \left(  \omega\right)  ^{2}+\left(
		\omega_{2}\right)  ^{2}\right]  \right)\\
		\ \lesssim\ &
		\frac{\min\{|k_1|,|k_2|,|k_3|,|k|\}}{|k||k_1||k_2|}\mathbf{1}_{\omega+\omega_1\ge \omega_2} \left(  \left\vert
		\phi^{\prime}\left(  \omega_{1}\right)  \left(  \omega-\omega
		_{2}\right)  \right\vert +C\left[  \left(  \omega\right)  ^{2}+\left(
		\omega_{2}\right)  ^{2}\right]  \right)\\
		\ \lesssim\ &
		\frac{\min\{|k_1|,|k_2|,|k_3|,|k|\}}{|k||k_1||k_2|}\mathbf{1}_{\omega+\omega_1\ge \omega_2} \left( 	|\phi^{\prime}|\left(  \omega_{1}\right) | \omega-\omega
		_{2}| +C\left[  \left(  \omega\right)  ^{2}+\left(
		\omega_{2}\right)  ^{2}\right]  \right),
	\end{aligned}
\end{equation}
for some constant $C>0$ depending only on $\phi$, $\phi'$ and $\phi''$. We can further bound \eqref{Lemma:GlobalExist:E4} as
\begin{equation}\label{Lemma:GlobalExist:E5}
	\begin{aligned}
		|\mathfrak{K}_\infty|
		\ \lesssim\  &
		\frac{\min\{|k_1|,|k_2|,|k_3|,|k|\}}{|k||k_1||k_2|}\mathbf{1}_{\omega+\omega_1\ge \omega_2} \left(     | \omega-\omega
		_{2}|	|\phi^{\prime}\left(  \omega_{1}\right)|  +   \left(  \omega\right)  ^{2}+\left(
		\omega_{2}\right)  ^{2} \right).
	\end{aligned}
\end{equation}
We can bound $\min\{|k_1|,|k_2|,|k_3|,|k|\}$ by $\min\{|k|,|k_2|\}$, that implies
\begin{equation}\label{Lemma:GlobalExist:E5}
	\begin{aligned}
		|\mathfrak{K}_\infty|
		\ \lesssim\  &
		\frac{1}{\max\{|k|,|k_2|\}|k_1|}\mathbf{1}_{\omega+\omega_1\ge \omega_2} \left(      | \omega-\omega
		_{2}|	|\phi^{\prime}\left(  \omega_{1}\right)| \ +\ \left(\max\{|k|,|k_2|\}\right)  ^{2\alpha'} \right).
	\end{aligned}
\end{equation}
We  estimate the second quantity on the right hand side of \eqref{Lemma:GlobalExist:E5} as  \begin{equation}\label{Lemma:GlobalExist:E6}	\frac{\left(\max\{|k|,|k_2|\}\right)  ^{2\alpha'}}{\max\{|k|,|k_2|\}|k_1|}\ \le\ \frac{\max\{|k|,|k_2|\}}{c_o}\ \le\ \frac{1}{c_o}.
\end{equation}
We now bound the first quantity on the  right hand side of \eqref{Lemma:GlobalExist:E5}
\begin{equation}\label{Lemma:GlobalExist:E7}
	\begin{aligned}
		\frac{1}{\max\{|k|,|k_2|\}|k_1|}      | \omega-\omega
		_{2}|	|\phi^{\prime}\left(  \omega_{1}\right)| \    \ \le \ & 	\frac{\max\{\omega,\omega_2\}}{\max\{|k|,|k_2|\}}    \frac{ \omega_{1}}{|k_1|} \Big|\frac{	\phi^{\prime}\left(  \omega_{1}\right) }{\omega_{1}} \Big|\\
		\ \le \ & \frac{\omega\big(\max\{|k|,|k_2|\}\big)}{\max\{|k|,|k_2|\}}    \frac{ \omega_{1}}{|k_1|}   \Big|\frac{	\phi^{\prime}\left(  \omega_{1}\right) }{\omega_{1}} \Big|.
	\end{aligned}
\end{equation}
By \eqref{Settings3}, we deduce that $\frac{\omega}{|k|}$ is bounded when $0\le |k|\le 1$ and when $\omega\in\mathrm{supp}(\phi)$. Moreover, since $ \frac{\lim_{\omega_1\to 0}	\phi^{\prime}\left(  \omega_{1}\right) }{\lim_{\omega_1\to 0}\omega_{1}}=\phi''(0)$,  the quantity $ \frac{ 	\phi^{\prime}\left(  \omega_{1}\right) }{ \omega_{1}}$ is also uniformly bounded. Therefore we can bound the first quantity on the  right hand side of \eqref{Lemma:GlobalExist:E5} by a constant, that in combination with \eqref{Lemma:GlobalExist:E6} and \eqref{Lemma:GlobalExist:E7} yields the boundedness of $	|\mathfrak{K}_\infty|$.

Finally we consider the case when all $|k|,|k_1|,|k_2|$ are close to $0$. We compute
\begin{equation}\begin{aligned}
		\mathfrak{K}_\infty
		\ =\ &\ \mathcal K_\infty  \Big[-\int_0^{\omega-\omega_2}\mathrm{d}\xi\int_0^{\omega_2-\omega_1}\mathrm{d}\xi'\phi''(\omega_1+\xi+\xi')\Big],
	\end{aligned}
\end{equation}
which can be bounded as

\begin{equation}\label{Lemma:GlobalExist:E8}
	\begin{aligned}
		|\mathfrak{K}_\infty|\ \leq\ &\
		\mathcal K_\infty  |\omega-\omega_2||\omega_2-\omega_1|\|\phi''\|_{L^\infty} 
		\ \lesssim  \
		\frac{\min\{|k_1|,|k_2|,|k|\}}{|k||k_1||k_2|}\mathbf{1}_{\omega+\omega_1\ge \omega_2} |\omega-\omega_2||\omega_2-\omega_1|\|\phi''\|_{L^\infty}.
	\end{aligned}
\end{equation}
We now consider $3$ cases.

{\it Case 1:  $|k_2|=\min\{|k_1|,|k_2|,|k|\}$.} We bound \eqref{Lemma:GlobalExist:E8} as
\begin{equation}\label{Lemma:GlobalExist:E9}
	\begin{aligned}
		|\mathfrak{K}_\infty|\ \leq\ &\
		\mathcal K_\infty  |\omega-\omega_2||\omega_2-\omega_1|\|\phi''\|_{L^\infty}\\
		\ \lesssim &\
		\frac{1}{|k||k_1|}\mathbf{1}_{\omega+\omega_1\ge \omega_2} |\omega-\omega_2||\omega_2-\omega_1|\|\phi''\|_{L^\infty}\\
		\ \lesssim &\
		\frac{\omega-\omega_2}{|k|}	\frac{\omega_1-\omega_2}{|k_1|}  \|\phi''\|_{L^\infty}\ \lesssim \ 	\frac{\omega}{|k|}	\frac{\omega_1}{|k_1|}  \|\phi''\|_{L^\infty},
	\end{aligned}
\end{equation}
where we  used the fact that since  $|k_2|=\min\{|k_1|,|k_2|,|k|\}$ then  $\omega-\omega_2, \omega_1-\omega_2\geq 0$. That  implies the boundedness of $	|\mathfrak{K}_\infty|$.

{\it Case 2:  $|k_1|=\min\{|k_1|,|k_2|,|k|\}$.} We bound \eqref{Lemma:GlobalExist:E8} as
\begin{equation}\label{Lemma:GlobalExist:E10}
	\begin{aligned}
		|\mathfrak{K}_\infty|
		\ \lesssim &\
		\frac{1}{|k||k_2|}\mathbf{1}_{\omega+\omega_1\ge \omega_2} |\omega-\omega_2||\omega_2-\omega_1|\|\phi''\|_{L^\infty}\\
		\ \lesssim &\
		\frac{|\omega-\omega_2|}{|k|}	\frac{\omega_2-\omega_1}{|k_2|}  \|\phi''\|_{L^\infty}\mathbf{1}_{\omega+\omega_1\ge \omega_2} \ \lesssim \ 	\frac{|\omega-\omega_2|}{|k|}	\frac{\omega_2}{|k_2|}  \|\phi''\|_{L^\infty}\mathbf{1}_{\omega+\omega_1\ge \omega_2}.
	\end{aligned}
\end{equation}
 If $\omega\ge \omega_2$, we bound $\frac{|\omega-\omega_2|}{|k|}\le \frac{\omega}{|k|}$, that implies the boundedness of $	|\mathfrak{K}_\infty|$. If $\omega< \omega_2$, we bound $\frac{|\omega-\omega_2|}{|k|}= \frac{\omega_2-\omega}{|k|}\le \frac{\omega_1}{|k|}\le \frac{\omega}{|k|}$, that also implies the boundedness of $	|\mathfrak{K}_\infty|$.

{\it Case 3:  $|k|=\min\{|k_1|,|k_2|,|k|\}$.} We bound \eqref{Lemma:GlobalExist:E8} as
\begin{equation}\label{Lemma:GlobalExist:E11}
	\begin{aligned}
		|\mathfrak{K}_\infty|
		\ \lesssim &\
		\frac{1}{|k_1||k_2|}\mathbf{1}_{\omega+\omega_1\ge \omega_2} |\omega-\omega_2||\omega_2-\omega_1|\|\phi''\|_{L^\infty}\\
		\ \lesssim &\
		\frac{-\omega+\omega_2}{|k_2|}	\frac{|\omega_2-\omega_1|}{|k_1|}  \|\phi''\|_{L^\infty}\mathbf{1}_{\omega+\omega_1\ge \omega_2} \ \lesssim \ 	\frac{|\omega_2-\omega_1|}{|k_1|}	\frac{\omega_2}{|k_2|}  \|\phi''\|_{L^\infty}\mathbf{1}_{\omega+\omega_1\ge \omega_2}.
	\end{aligned}
\end{equation}
If $\omega_1\ge \omega_2$, we bound $\frac{|\omega_2-\omega_1|}{|k_1|}\le \frac{\omega_1}{|k_1|}$, that implies the boundedness of $	|\mathfrak{K}_\infty|$. If $\omega_1< \omega_2$, we bound $\frac{|\omega_2-\omega_1|}{|k_1|}= \frac{\omega_2-\omega_1}{|k_1|}\le \frac{\omega}{|k_1|}\le \frac{\omega_1}{|k_1|}$, that also implies the boundedness of $	|\mathfrak{K}_\infty|$.

{\it Step 3: Global existence for \eqref{Lemma:GlobalExist:E3}, with $n<\infty$.} 

We define the operator $\mathscr{T}[g] $ in the weak form as follow

\begin{equation}\label{Lemma:GlobalExist:E12}
	\begin{aligned}
		\left\langle \mathscr{T}[g],  \phi \right\rangle\ = \ &	\int_{\mathbb{R}_+}\mathrm{d}\omega \mathscr{T}[g]  \phi \\
		\ =\ & C_Q\iiint_{\mathbb{R}_+^{3}}\mathrm{d}\omega_1\,\mathrm{d}\omega_2\,\mathrm{d}\omega g_1g_2g\mathcal K_n  [-\phi(\omega)-\phi(\omega_1)+\phi(\omega_2)+\phi(\omega+\omega_1-\omega_2)].
	\end{aligned}
\end{equation}
The operator $\mathscr{T}$ is actually  bounded since we can estimate
\begin{equation}\label{Lemma:GlobalExist:E13}
	\begin{aligned}
		\sup_{\|\phi\|_{L^\infty}=1}	\Big|\left\langle \mathscr{T}[g],  \phi \right\rangle\Big|\ \le \ & \sup_{\|\phi\|_{L^\infty}=1} C_Q	\Big|\iiint_{\mathbb{R}_+^{3}}\mathrm{d}\omega_1\,\mathrm{d}\omega_2\,\mathrm{d}\omega g_1g_2g\mathcal K_n  \\
		&\times [-\phi(\omega)-\phi(\omega_1)+\phi(\omega_2)+\phi(\omega+\omega_1-\omega_2)]	\Big|\\
		\lesssim \ & \Big|\iiint_{\mathbb{R}_+^{3}}\mathrm{d}\omega_1\,\mathrm{d}\omega_2\,\mathrm{d}\omega |g_1g_2g|	\Big| \lesssim \   \Big|\int_{\mathbb{R}_+}\mathrm{d}\omega |g|	\Big|^3.
	\end{aligned}
\end{equation}
Next, we show that $\mathscr{T}$ is Lipschitz in the set  $S_{\mathscr{T}}:=\Big\{g ~~| ~~ g\ge0 , \|g\|_{L^1}\le \mathscr{M}\Big\}$. To this end, for $g,g^o\in S_{\mathscr{T}}$, we bound

\begin{equation}\label{Lemma:GlobalExist:E14}
	\begin{aligned}
		\sup_{\|\phi\|_{L^\infty}=1}	\Big|\left\langle \mathscr{T}[g]-\mathscr{T}[g^o],  \phi \right\rangle\Big|\ \le \ & \sup_{\|\phi\|_{L^\infty}=1} C_Q	\Big|\iiint_{\mathbb{R}_+^{3}}\mathrm{d}\omega_1\,\mathrm{d}\omega_2\,\mathrm{d}\omega \big[g_1g_2g-g^o_1g^o_2g^o\big]\mathcal K_n  \\
		&\times [-\phi(\omega)-\phi(\omega_1)+\phi(\omega_2)+\phi(\omega+\omega_1-\omega_2)]	\Big|\\
		\lesssim \ & \Big|\iiint_{\mathbb{R}_+^{3}}\mathrm{d}\omega_1\,\mathrm{d}\omega_2\,\mathrm{d}\omega \big[g_1g_2g-g^o_1g^o_2g^o\big]	\Big| \lesssim \  \mathscr{M}^2 \|g-g^o\|_{L^1}.
	\end{aligned}
\end{equation}

The local existence of a solution   $g^n\in C^1([0,T],L^1(\mathbb{R}_+))$ to \eqref{Lemma:GlobalExist:E3} then follows by a standard fixed point argument for a small time  $T>0$ depending on $ \mathscr{M}$. By the conservation \eqref{Lemma:GlobalExist:E3b}, the argument can then be iterated from $[0,T]$ to $[T,2T]$, $[2T,3T]$, giving a global existence of solutions in  $g^n\in C^1([0,\infty),L^1(\mathbb{R}_+))$.

{\it Step 4: Global existence for \eqref{Lemma:GlobalExist:E3}, with $n=\infty$.} 

To obtain a global existence result for \eqref{Lemma:GlobalExist:E3}, with $n=\infty$, we will consider the sequences of solutions obtained in Step 3 $\{g^m\}\subset C^1([0,\infty),L^1(\mathbb{R}_+))$. We will now estimate the difference
$$\int_{\left[  0,\infty\right)  }\mathrm d\omega g^m\left(  t_{2},\omega\right)
\phi\left(  \omega\right) -\int_{\left[  0,\infty\right)
}\mathrm d\omega g^m\left(  t_{1},\omega\right)  \phi\left(  \omega\right) $$
for any $\phi\in C^{2}_c\left(  \left[  0,\infty\right)  \right)$, 
$t_{1},t_{2}\in\left[  0,\infty\right)  .$ Using the uniform boundedness of $|\mathfrak{K}_n|$, we obtain:%
\[
\left\vert \int_{\left[  0,\infty\right)  }\mathrm d\omega g^m\left(  t_{2}%
,\omega\right)  \phi\left(  \omega\right) -\int_{\left[
	0,\infty\right)  }\mathrm d\omega g^m\left(  t_{1},\omega\right)  \phi\left(
\omega\right)  d\omega\right\vert \leq C\left\vert t_{2}-t_{1}\right\vert
\]
where $C>0$ is independent of $m.$ By  Arzela-Ascoli's theorem, there exists a  subsequence $\left\{  i_{m}\right\}  $ such that  $$g^{i_{m}}\rightharpoonup g^\infty\in C^1([0,\infty),L^1(\mathbb{R}_+)).$$
The function $g^\infty$ is a mild solution of	 \eqref{Lemma:GlobalExist:E3}, with $n=\infty$.
\end{proof}
\section{Transfer of energy towards large values of wavenumbers}

\begin{lemma}
\label{Lemma:Monotonicity} Assume  $\phi\in C_c^2\left([0,\infty)\right) $ and  is a convex function. Let $\beth	\in C\left( \mathbb{R}^3\right) $ and  $\beth	(k)=\phi(\omega_k)$. Then the following identity holds  
\begin{equation}	\label{Lemma:Monotonicity:1} 
	\begin{aligned}
		\int_{{\mathbb{R}}^{3}}\mathrm dk\mathcal Q \left[ f\right](k)  \beth	(k)\
		\ge 0.	\end{aligned}
\end{equation}

\end{lemma}

\begin{proof}

We   define

\begin{align*}
	\omega_{Max}\left( \omega,\omega_{1},\omega_{2}\right) & =\max\left\{\omega,\omega_{1},\omega_{2}\right\} , \ 
	\omega_{Min}\left( \omega,\omega_{1},\omega_{2}\right)   =\min\left\{ \omega,\omega_{1},\omega_{2}\right\} , \\
	\omega_{Mid}\left( \omega,\omega_{1},\omega_{2}\right) & =\omega_{j}\in\left\{ \omega,\omega_{1},\omega_{2}\right\} \backslash\{\omega_{Max},\omega_{Min}\},
\end{align*}
with $j\in\left\{ 1,2,3\right\} .$ Moreover, we suppose that $|k_{Max}|$ is associated to $\omega_{Max}$, $|k_{Min}|$ is associated to $\omega_{Min}$, $|k_{Mid}|$ is associated to $\omega_{Mid}$.

By a symmetry argument, we rewrite \eqref{Lemma:TestFunction:E6} as 
\begin{equation}\label{Lemma:TestFunction:E7}
	\begin{aligned}
		&	\int_{\mathbb{R}_+}\mathrm{d}\omega \mathcal Q \left[ f\right]  \phi(\omega)|k| \mho\\
		\ =\ & C_Q\iiint_{\mathbb{R}_+^{3}}\mathrm{d}\omega_1\,\mathrm{d}\omega_2\,\mathrm{d}\omega f_1f_2f\mathbf{1}_{\omega+\omega_1-\omega_2\ge 0}\\
		&\times 	\frac{1}{3}\Big\{ [-\phi(\omega_{Max})-\phi(\omega_{Min})+\phi(\omega_{Mid})+\phi(\omega_{Max}+\omega_{Min}-\omega_{Mid})]\\
		&\ \ \ \ \times\Xi(\omega_{Max},\omega_{Min},\omega_{Mid},\omega_{Max}+\omega_{Min}-\omega_{Mid})  \\
		&+[-\phi(\omega_{Max})-\phi(\omega_{Mid})+\phi(\omega_{Min})+\phi(\omega_{Max}+\omega_{Mid}-\omega_{Min})]\\
		&\ \ \ \ \times\Xi(\omega_{Max},\omega_{Mid},\omega_{Min},\omega_{Max}+\omega_{Mid}-\omega_{Min})  \\
		&+[-\phi(\omega_{Min})-\phi(\omega_{Mid})+\phi(\omega_{Max})+\phi(\omega_{Min}+\omega_{Mid}-\omega_{Max})]\\
		&\ \ \ \ \times\Xi(\omega_{Min},\omega_{Mid},\omega_{Max},\omega_{Min}+\omega_{Mid}-\omega_{Max})\Big\}\\
		\ =\ & C_Q\iiint_{\mathbb{R}_+^{3}}\mathrm{d}\omega_1\,\mathrm{d}\omega_2\,\mathrm{d}\omega f_1f_2f\mathbf{1}_{\omega+\omega_1-\omega_2\ge 0}\\
		&\times 	\frac{1}{3}\Big\{ [-\phi(\omega_{Max})-\phi(\omega_{Min})+\phi(\omega_{Mid})+\phi(\omega_{Max}+\omega_{Min}-\omega_{Mid})]\\
		&\ \ \ \ \times\mho(\omega_{Max})\mho(\omega_{Min})\mho(\omega_{Mid})\mho(\omega_{Max}+\omega_{Min}-\omega_{Mid})|k_{Min}|  \\
		&+[-\phi(\omega_{Max})-\phi(\omega_{Mid})+\phi(\omega_{Min})+\phi(\omega_{Max}+\omega_{Mid}-\omega_{Min})]\\
		&\ \ \ \ \times\mho(\omega_{Max})\mho(\omega_{Min})\mho(\omega_{Mid})\mho(\omega_{Max}+\omega_{Mid}-\omega_{Min})|k_{Min}|  \\
		&+[-\phi(\omega_{Min})-\phi(\omega_{Mid})+\phi(\omega_{Max})+\phi(\omega_{Min}+\omega_{Mid}-\omega_{Max})]\\
		&\ \ \ \ \times\Xi(\omega_{Min},\omega_{Mid},\omega_{Max},\omega_{Min}+\omega_{Mid}-\omega_{Max})\mathbf{1}_{\omega_{Min}+\omega_{Mid}-\omega_{Max}\ge 0}\Big\}.
	\end{aligned}
\end{equation}
		Next, we will show that \begin{equation}\label{Lemma:TestFunction:E9}-\phi(\omega_{Min})-\phi(\omega_{Mid})+\phi(\omega_{Max})+\phi(\omega_{Min}+\omega_{Mid}-\omega_{Max})\ge 0,\end{equation}
	which can be seen by a simple proof from the convexity of $\phi$. We recall this standard proof below.
	Define the function $$\varphi\left( x\right) =\phi\left( 
	\frac{\omega_{Min}+\omega_{Mid}}{2}+x\right) +\phi\left( \frac{\omega_{Min}+\omega_{Mid}}{2}-x\right) $$ in which  $x\geq0,$ and $\frac{\omega_{Min}+\omega_{Mid}}{2}-x>0.$ We first consider the case when  $\phi\in C^{2}(\mathbb{R}_+)$, and obtain 
	$
	\varphi^{\prime}\left( 0\right)   = 0,  
	\varphi^{\prime\prime}\left( x\right)   \geq0.$
	We then deduce 
	$\varphi^{\prime}\left( x\right) \geq0$ when $x\geq0,$ yielding $-\phi(\omega_{Min})-\phi(\omega_{Mid})+\phi(\omega_{Max})+\phi(\omega_{Min}+\omega_{Mid}-\omega_{Max})\ge 0.$  The case when  $\phi$ is in $C(\mathbb{R}_+)$  can be proved by  performing an approximation argument. We then deduce
		\begin{equation}\label{Lemma:TestFunction:E7a}
		\begin{aligned}
			&	\int_{\mathbb{R}_+}\mathrm{d}\omega \mathcal Q \left[ f\right]  \phi(\omega)|k| \mho\\
			\ \ge \ & C_Q\iiint_{\mathbb{R}_+^{3}}\mathrm{d}\omega_1\,\mathrm{d}\omega_2\,\mathrm{d}\omega f_1f_2f\mathbf{1}_{\omega+\omega_1-\omega_2\ge 0}\\
			&\times 	\frac{1}{3}\Big\{ [-\phi(\omega_{Max})-\phi(\omega_{Min})+\phi(\omega_{Mid})+\phi(\omega_{Max}+\omega_{Min}-\omega_{Mid})]\\
			&\ \ \ \ \times\mho(\omega_{Max})\mho(\omega_{Min})\mho(\omega_{Mid})\mho(\omega_{Max}+\omega_{Min}-\omega_{Mid})|k_{Min}|  \\
			&+[-\phi(\omega_{Max})-\phi(\omega_{Mid})+\phi(\omega_{Min})+\phi(\omega_{Max}+\omega_{Mid}-\omega_{Min})]\\
			&\ \ \ \ \times\mho(\omega_{Max})\mho(\omega_{Min})\mho(\omega_{Mid})\mho(\omega_{Max}+\omega_{Mid}-\omega_{Min})|k_{Min}|\Big\}.
		\end{aligned}
	\end{equation}
	
	By our assumption,    $\phi$ is  convex, then 
	\begin{equation}
	\phi(\omega_{Max}+\omega_{Min}-\omega_{Mid}) \ +\ \phi(\omega_{Max}+\omega_{Mid}-\omega_{Min}) 	 \geq 2 \phi(\omega_{Max}),  
	\end{equation}
	
	yielding
		\begin{equation}\begin{aligned}
	&	-\phi(\omega_{Max})-\phi(\omega_{Min})+\phi(\omega_{Mid})+\phi(\omega_{Max}+\omega_{Min}-\omega_{Mid})\\
		\geq \ & -[-\phi(\omega_{Max})-\phi(\omega_{Mid})+\phi(\omega_{Min})+\phi(\omega_{Max}+\omega_{Mid}-\omega_{Min})].  \end{aligned}
		\end{equation}
	Multiplying the above inequality with $\mho(\omega_{Max}+\omega_{Min}-\omega_{Mid})$, we find
	\begin{equation}\begin{aligned}
			&	[-\phi(\omega_{Max})-\phi(\omega_{Min})+\phi(\omega_{Mid})+\phi(\omega_{Max}+\omega_{Min}-\omega_{Mid})]\mho(\omega_{Max}+\omega_{Min}-\omega_{Mid})\\
			\geq \ & -[-\phi(\omega_{Max})-\phi(\omega_{Mid})+\phi(\omega_{Min})+\phi(\omega_{Max}+\omega_{Mid}-\omega_{Min})]\mho(\omega_{Max}+\omega_{Min}-\omega_{Mid}),  \end{aligned}
	\end{equation}
leading to 
\begin{equation}\begin{aligned}
		&	[-\phi(\omega_{Max})-\phi(\omega_{Min})+\phi(\omega_{Mid})+\phi(\omega_{Max}+\omega_{Min}-\omega_{Mid})]\mho(\omega_{Max}+\omega_{Min}-\omega_{Mid})
		\\
		& +  [-\phi(\omega_{Max})-\phi(\omega_{Mid})+\phi(\omega_{Min})+\phi(\omega_{Max}+\omega_{Mid}-\omega_{Min})]\mho(\omega_{Max}+\omega_{Mid}-\omega_{Min})\\
		\geq \ & [-\phi(\omega_{Max})-\phi(\omega_{Mid})+\phi(\omega_{Min})+\phi(\omega_{Max}+\omega_{Mid}-\omega_{Min})]\\
		&\times[\mho(\omega_{Max}+\omega_{Mid}-\omega_{Min})-\mho(\omega_{Max}+\omega_{Min}-\omega_{Mid})].  \end{aligned}
\end{equation}
The same proof of \eqref{Lemma:TestFunction:E9} gives $-\phi(\omega_{Max})-\phi(\omega_{Mid})+\phi(\omega_{Min})+\phi(\omega_{Max}+\omega_{Mid}-\omega_{Min})\ge 0$. As thus, we obtain
\begin{equation}\begin{aligned}
		&	[-\phi(\omega_{Max})-\phi(\omega_{Min})+\phi(\omega_{Mid})+\phi(\omega_{Max}+\omega_{Min}-\omega_{Mid})]\mho(\omega_{Max}+\omega_{Min}-\omega_{Mid})
		\\
		& +  [-\phi(\omega_{Max})-\phi(\omega_{Mid})+\phi(\omega_{Min})+\phi(\omega_{Max}+\omega_{Mid}-\omega_{Min})]\mho(\omega_{Max}+\omega_{Mid}-\omega_{Min})\ 
		\geq \ 0,  \end{aligned}
\end{equation}
which, in combination with \eqref{Lemma:TestFunction:E7a}, implies
\begin{equation}\label{Lemma:TestFunction:E8}
	\begin{aligned}
			\int_{\mathbb{R}_+}\mathrm{d}\omega \mathcal Q \left[ f\right]  \phi(\omega)|k| \mho
		\ \ge \ & 0.
	\end{aligned}
\end{equation}

\end{proof}

\begin{proof}

	[Proof of the Theorem \ref{Theorem:Main}] 
	The global existence follows from Proposition \ref{Propo:GlobalWellPosed}.  The energy conservation \eqref{Theorem:Main:2} is straightforward. We will now prove the energy cascade \eqref{Theorem:Main:3}. Without loss of generality, we assume that the mass is always $1$ $$\int_{\mathbb{R}^3}\mathrm{d}k f(t,k)=\int_{\mathbb{R}_+}\mathrm{d}\omega f(t,\omega)|k|\mho=1.$$ 
	First, from \eqref{Lemma:TestFunction:E7}, we deduce
		\begin{equation}\label{Lemma:Theorem:E1}
		\begin{aligned}
			& 	\int_{\mathbb{R}_+}\mathrm{d}\omega\partial_t f  \phi(\omega)|k| \mho\\
			\ =\ & C_Q\iiint_{\mathbb{R}_+^{3}}\mathrm{d}\omega_1\,\mathrm{d}\omega_2\,\mathrm{d}\omega f_1f_2f\Xi(\omega,\omega_1,\omega_2,\omega+\omega_1-\omega_2)\mathbf{1}_{\omega+\omega_1-\omega_2\ge 0}\\
			&\times 	\frac{1}{3}\Big\{ [-\phi(\omega_{Max})-\phi(\omega_{Min})+\phi(\omega_{Mid})+\phi(\omega_{Max}+\omega_{Min}-\omega_{Mid})]\\
			&\ \ \ \ \times\mho(\omega_{Max})\mho(\omega_{Min})\mho(\omega_{Mid})\mho(\omega_{Max}+\omega_{Min}-\omega_{Mid})|k_{Min}|  \\
			&+[-\phi(\omega_{Max})-\phi(\omega_{Mid})+\phi(\omega_{Min})+\phi(\omega_{Max}+\omega_{Mid}-\omega_{Min})]\\
			&\ \ \ \ \times\mho(\omega_{Max})\mho(\omega_{Min})\mho(\omega_{Mid})\mho(\omega_{Max}+\omega_{Mid}-\omega_{Min})|k_{Min}|  \\
			&+[-\phi(\omega_{Min})-\phi(\omega_{Mid})+\phi(\omega_{Max})+\phi(\omega_{Min}+\omega_{Mid}-\omega_{Max})]\\
			&\ \ \ \ \times\Xi(\omega_{Min},\omega_{Mid},\omega_{Max},\omega_{Min}+\omega_{Mid}-\omega_{Max})\mathbf{1}_{\omega_{Min}+\omega_{Mid}-\omega_{Max}\ge 0}\Big\},
		\end{aligned}
	\end{equation}
	in which  $\phi\in C^{2}\left(  \left[  0,\infty\right)  \right)  $ is a test
	function satisfying the following properties
		\begin{equation}\label{Lemma:Theorem:E0}\text{for any }\omega\in\left[  0,\infty\right):\,\,\,\phi\left(  \omega\right)  \geq0,\ \ \ \ \phi^{\prime\prime}\left(  \omega\right)
		<-C_\phi<0.	\end{equation}
We first observe that 
	\begin{equation}\label{Lemma:Concave:1}
		\begin{aligned}
			&[-\phi(\omega_{Min})-\phi(\omega_{Mid})+\phi(\omega_{Max})+\phi(\omega_{Min}+\omega_{Mid}-\omega_{Max})]\\
			&\ \ \ \ \times\Xi(\omega_{Min},\omega_{Mid},\omega_{Max},\omega_{Min}+\omega_{Mid}-\omega_{Max})\Big\}\\	=\ &\int_{0}^{\omega_{Max}-\omega_{Min}}\mathrm{d}t_1\int_{0}^{\omega_{Max}-\omega_{Mid}}\mathrm{d}t_2\phi''(t_1+t_2+\omega_{Min})\\
			&\ \ \ \ \times\Xi(\omega_{Min},\omega_{Mid},\omega_{Max},\omega_{Min}+\omega_{Mid}-\omega_{Max})\ \le 0.
		\end{aligned}
	\end{equation}

Note that in our definition of mild solutions, $\phi$ is chosen to be in $C_c^2([0,\infty))$ but we can choose $\phi$ as above by using an approximated sequence.

	Now, we have
	\begin{equation}\label{Lemma:Concave:2}
		\begin{aligned}
			& 	  [-\phi(\omega_{Max})-\phi(\omega_{Min})+\phi(\omega_{Mid})+\phi(\omega_{Max}+\omega_{Min}-\omega_{Mid})]\mho(\omega_{Max}+\omega_{Min}-\omega_{Mid})  \\
			&+[-\phi(\omega_{Max})-\phi(\omega_{Mid})+\phi(\omega_{Min})+\phi(\omega_{Max}+\omega_{Mid}-\omega_{Min})]\mho(\omega_{Max}+\omega_{Mid}-\omega_{Min})\\
			=\		& 	  -\int_{0}^{\omega_{Mid}-\omega_{Min}}\mathrm{d}t\int_{0}^{\omega_{Max}-\omega_{Mid}}\mathrm{d}t_0\mho(\omega_{Max}+\omega_{Min}-\omega_{Mid}) \phi''(\omega_{Min}+t+t_0) \\
			&+\int_{0}^{\omega_{Mid}-\omega_{Min}}\mathrm{d}t\int_{0}^{\omega_{Max}-\omega_{Min}}\mathrm{d}t_0\mho(\omega_{Max}-\omega_{Min}+\omega_{Mid}) \phi''(\omega_{Min}+t+t_0)\\
			=\		& 	  \int_{0}^{\omega_{Mid}-\omega_{Min}}\mathrm{d}t\int_{0}^{\omega_{Max}-\omega_{Mid}}\mathrm{d}t_0 \phi''(\omega_{Min}+t+t_0) \\
			&\times[\mho(\omega_{Max}-\omega_{Min}+\omega_{Mid})-\mho(\omega_{Max}+\omega_{Min}-\omega_{Mid})]\\
			&+\int_{0}^{\omega_{Mid}-\omega_{Min}}\mathrm{d}t\int_{0}^{\omega_{Mid}-\omega_{Min}}\mathrm{d}t_0\mho(\omega_{Max}-\omega_{Min}+\omega_{Mid}) \phi''(\omega_{Min}+t+t_0)\\
			\le\ &	\int_{0}^{\omega_{Mid}-\omega_{Min}}\mathrm{d}t\int_{0}^{\omega_{Mid}-\omega_{Min}}\mathrm{d}t_0\mho(\omega_{Max}-\omega_{Min}+\omega_{Mid}) \phi''(\omega_{Min}+t+t_0)\\	\le\ & -[\omega_{Mid}-\omega_{Min}]^2C_\phi {C}_\mho^2\mho(\omega_{Max}-\omega_{Min}+\omega_{Mid}) .
		\end{aligned}
	\end{equation}

	Let $\delta$ be a small constant and $\rho$ be a large constant. We suppose $\omega_{Min}\in\left[  0,\delta^{2}\right]  $ and $\omega_{Mid}%
	,\omega_{Max}\in\left(  2\delta^{2},\frac{20\rho}{\delta}\right]  .$ Then by \eqref{Lemma:Concave:1} and \eqref{Lemma:Concave:2} 
		\begin{equation}\label{Lemma:Theorem:E7}
		\begin{aligned}
			& 	\frac{1}{3}\Big\{ [-\phi(\omega_{Max})-\phi(\omega_{Min})+\phi(\omega_{Mid})+\phi(\omega_{Max}+\omega_{Min}-\omega_{Mid})]\\
			&\ \ \ \ \times\mho(\omega_{Max})\mho(\omega_{Min})\mho(\omega_{Mid})\mho(\omega_{Max}+\omega_{Min}-\omega_{Mid})|k_{Min}|  \\
			&+[-\phi(\omega_{Max})-\phi(\omega_{Mid})+\phi(\omega_{Min})+\phi(\omega_{Max}+\omega_{Mid}-\omega_{Min})]\\
			&\ \ \ \ \times\mho(\omega_{Max})\mho(\omega_{Min})\mho(\omega_{Mid})\mho(\omega_{Max}+\omega_{Mid}-\omega_{Min})|k_{Min}|  \\
			&+[-\phi(\omega_{Min})-\phi(\omega_{Mid})+\phi(\omega_{Max})+\phi(\omega_{Min}+\omega_{Mid}-\omega_{Max})]\\
			&\ \ \ \ \times\Xi(\omega_{Min},\omega_{Mid},\omega_{Max
			},\omega_{Min}+\omega_{Mid}-\omega_{Max})\mathbf{1}_{\omega_{Min}+\omega_{Mid}-\omega_{Max}\ge 0}\Big\}\\
		\le\ 	& -	\mathfrak{C}_2\delta^{4}\mho(\delta^2) ,
		\end{aligned}
	\end{equation}
	where $\mathfrak{C}_2$  is a universal constant, yielding
	
	\begin{equation}\label{Lemma:Theorem:E8}
	\int_{\mathbb{R}_+}\mathrm{d}\omega\partial_t f (t,\omega) \phi(\omega)|k| \mho  \leq-\mathfrak{C}_2\delta^{4}\mho(\delta^2) \left(  \int_{\left(
		2\delta^{2},\frac{20\rho}{\delta}\right]  }\mathrm{d}\omega f(t,\omega) 
	\right)  ^{2}\int_{\left[  0,\delta^{2}\right]  }\mathrm{d}\omega f(t,\omega).
	\end{equation}
	By Proposition \ref{Lemma:Support}, there exists a time $t'$ such that 
		\begin{equation}\label{Lemma:Theorem:E9}
		\int_{\left[  0,\delta^{2}/2\right]  }\mathrm{d}\omega f(t',\omega)|k|\mho  \geq
		\mathfrak{C}_3\ \ , 
	\end{equation}
where $\mathfrak{C}_3>0$ is a constant. Using the  test function $
\beth	(k)=\phi(\omega_k) =\left(   \delta^2  - \omega_k\right)  _{+}%
$ in Lemma \ref{Lemma:Monotonicity}, we find
\begin{equation}\label{Lemma:Theorem:E12:b}
	\begin{aligned}
&	\delta^2\int_{0}^{\delta^2}\mathrm d\omega f\left(  t,\omega\right)   
	 |k| \mho \ \geq \	\int_{0}^{\infty}\mathrm d\omega f\left(  t,\omega\right)  \left(   \delta^2 - \omega\right)
	_{+}|k| \mho \\
\ \geq\ 	 &\int_{0}^{\infty}\mathrm d\omega f\left( t',\omega\right)   \left( \delta^2 -\omega\right)  _{+} |k| \mho\ \geq\ \int_{0}^{\delta^2/2}\mathrm d\omega f\left( t',\omega\right)   \left( \delta^2 -\omega\right)  _{+} |k| \mho\\ \geq\ & \frac{\delta^2}{2} \int_{0}^{\delta^2/2}\mathrm d\omega f\left( t',\omega\right)  |k| \mho,\end{aligned}
\end{equation} leading to
	\begin{equation}\label{Lemma:Theorem:E10}
	\int_{\left[  0,\delta^{2}\right]  }\mathrm{d}\omega f(t,\omega)\  \geq \	\mathfrak{C}_4\int_{\left[  0,\delta^{2}/2\right]  }\mathrm{d}\omega f(t',\omega)|k|\mho \ \geq\
	\mathfrak{C}_5\ \ ,\ \ t\geq t',
\end{equation}
for some constants $\mathfrak{C}_4,\mathfrak{C}_5>0$.

 Inequality \eqref{Lemma:Theorem:E10}, in combination with \eqref{Lemma:Theorem:E8}, implies
	\begin{equation}\label{Lemma:Theorem:E11}
	\int_{\mathbb{R}_+}\mathrm{d}\omega\partial_t f(t,\omega)  \phi(\omega)|k| \mho  \leq-\mathfrak{C}_6\delta^{4}\mho(\delta^2) \left(  \int_{\left(
		2\delta^{2},\frac{20\rho}{\delta}\right]  }\mathrm{d}\omega f(t,\omega) 
	\right)  ^{2}\ \ ,\ \ t\geq t',
\end{equation}
for some constant  $\mathfrak{C}_6>0$.

Next, we use the  test function $
\beth	(k)=\phi(\omega_k) =\left(  1-\Theta\omega_k\right)  _{+}%
$ in Lemma \ref{Lemma:Monotonicity}
	where $\Theta>0$ will be specified later. It then follows from inequality \eqref{Lemma:Monotonicity:1} that
	\begin{equation}\label{Lemma:Theorem:E12}
		\int_{0}^{\infty}\mathrm d\omega f\left(  t,\omega\right)  \left(  1-\Theta\omega\right)
		_{+}|k| \mho  \geq\int_{0}^{\infty}\mathrm d\omega f\left( 0,\omega\right)   \left(
		1-\Theta\omega\right)  _{+} |k| \mho.
	\end{equation}
	Let us choose $\Theta \rho\leq1$  
	\begin{equation}\label{Lemma:Theorem:E13}
	\int_{0}^{\frac{1}{\Theta}}\mathrm d\omega f\left(  t,\omega\right) |k| \mho    \geq \left(  1-\Theta \rho \right)\int_{0}%
	^{\rho}\mathrm d\omega f\left(0,\omega\right) |k| \mho.
	\end{equation}
We assume that $\rho$ is large enough such that 
	\begin{equation}\label{Lemma:Theorem:E14}
	 \int_{0}%
	^{\rho}\mathrm d\omega f\left(0,\omega\right) |k| \mho\ \ge  \ 1-\frac{\delta}{20},
\end{equation}
	which, in combination with \eqref{Lemma:Theorem:E13}, implies
	\begin{equation}\label{Lemma:Theorem:E15}
	\int_{0}^{\frac{1}{\Theta}}\mathrm d\omega f\left(  t,\omega\right) |k| \mho    \geq \left(  1-\Theta \rho \right)\left( 1-\frac{\delta}{20} \right)\ \ge \ 1-\Theta \rho -\frac{\delta}{20}.
	\end{equation}
By setting $\Theta=\frac{\delta}{20\rho}$, we obtain
	\begin{equation}\label{Lemma:Theorem:E16}
	\int_{0}^{\frac{20\rho}{\delta}}\mathrm d\omega f\left(  t,\omega\right) |k| \mho\  \ge  \ 1 -\frac{\delta}{10}.
\end{equation}
	Next, we will prove that there exists $t''>t'$ such that 
	\begin{equation}\label{Lemma:Theorem:E17}
		\int^{2\delta^2}_0\mathrm d\omega f\left(  t'',\omega\right) |k| \mho  \ \ge\   1 -\frac{\delta}{20}.
	\end{equation}
	Suppose the contrary that 
	\begin{equation}\label{Lemma:Theorem:E18}
		\int^{2\delta^2}_0\mathrm d\omega f\left(  t,\omega\right) |k| \mho  \ \le\   1 -\frac{\delta}{20}\ \ ,\ \ \forall t> t',
	\end{equation}
	which, by \eqref{Lemma:Theorem:E16}, yields
		\begin{equation}\label{Lemma:Theorem:E19}
		\int_{2\delta^2}^{\frac{20\rho}{\delta}}\mathrm d\omega f\left(  t,\omega\right) |k| \mho\  \ge  \ \frac{\delta}{20}\ \ ,\ \ \forall t> t'.
	\end{equation}
	Combining \eqref{Lemma:Theorem:E11} and \eqref{Lemma:Theorem:E19}, we find
	\begin{equation}\label{Lemma:Theorem:E20}
	\int_{\mathbb{R}_+}\mathrm{d}\omega\partial_t f(t,\omega)  \phi(\omega)|k| \mho  \leq-\mathfrak{C}_7\delta^{6}\mho(\delta^2) , \ \ \forall t> t',
\end{equation}
yielding a contradiction. As a result,  \eqref{Lemma:Theorem:E17} is proved, by the contradiction argument.

	Applying again   \eqref{Lemma:Theorem:E13}, we obtain 
 
	$$\int_{\left[  0,\frac{20\delta^2}{\delta}\right)
	} \mathrm d\omega f\left(t,\omega\right) |k| \mho \ \geq\  \left(  1-\frac{\delta}{20\delta^2} 2\delta^2 \right)\int_{0}%
^{2\delta^2}\mathrm d\omega f\left(t'',\omega\right) |k| \mho $$ $$\geq\left(1 -\frac{\delta}{20}\right)\left(1 -\frac{\delta}{10}\right)\geq 1-\frac{3\delta}{20},$$
	for any $t>t''.$ This implies
	\begin{equation}\label{Lemma:Theorem:E21}
		\int_{ [{20\delta },\infty)  
		} d\omega f\left(t,\omega\right) |k| \mho \ \leq\  \frac{3\delta}{20},
	\end{equation}
	for any $t>t''.$ 
	
Next, let $\varphi$ be any bounded test function in $C([0,\infty])$, we bound for $t>t''$ 

	\begin{equation}\label{Lemma:Theorem:E22}\begin{aligned}
&	\left\vert \int_ {[0,\infty) }\mathrm d\omega f\left(t,\omega\right) |k| \mho  \varphi\left(  \omega\right)
	-\varphi\left(  0\right)  \right\vert\\
	\le &  \left\vert \int_ {[20\delta,\infty) } \mathrm d\omega f\left(t,\omega\right) |k| \mho  \varphi\left(  \omega\right)
	-\int_ {[20\delta,\infty) }\mathrm d\omega f\left(t,\omega\right) |k| \mho\varphi\left(  0\right)  \right\vert\\
	& \ + \  \left\vert \int_ {[0,20\delta) }\mathrm d\omega f\left(t,\omega\right) |k| \mho  \varphi\left(  \omega\right)
	-\int_ {[0,20\delta)}\mathrm d\omega f\left(t,\omega\right) |k| \mho\varphi\left(  0\right)  \right\vert\\
\le 	&   \sup_{\omega\in\left[
		0,20\delta \right]  }\left\vert \varphi\left(  \omega\right)
	-\varphi\left(  0\right)  \right\vert +\frac{3\delta}{20}\sup_{\omega
	}\left\vert \varphi\left(  \omega\right)
-\varphi\left(  0\right)  \right\vert\\
\le 	&   \sup_{\omega\in\left[
	0,20\delta \right]  }\left\vert \varphi\left(  \omega\right)
-\varphi\left(  0\right)  \right\vert +\frac{3\delta}{10}\sup_{\omega
}\left\vert \varphi\left(  \omega\right)\right\vert,\end{aligned}
	\end{equation}
	which implies
		\begin{equation}\label{Lemma:Theorem:E23}\begin{aligned}
			&\lim_{t\to\infty}	\left\vert \int_ {[0,\infty) } \mathrm d\omega f\left(t,\omega\right) |k| \mho  \varphi\left(  \omega\right)
			-\varphi\left(  0\right)  \right\vert \ = \ 0.\end{aligned}
	\end{equation}
	Inequality \eqref{Lemma:Theorem:E22} also implies that for any $\mathfrak R>0$, 
		\begin{equation}\label{Lemma:Theorem:E24}\begin{aligned}
			&\lim_{t\to\infty}\int_{0}^\infty	\mathrm  d\omega f\left(t,\omega\right) |k| \mho  \ \omega \chi_{\{|k|\in[0,\mathfrak R]\}}
			  \ = \ 0,\end{aligned}
	\end{equation}
yielding
	\begin{equation}\label{Lemma:Theorem:E25}\begin{aligned}
		&\lim_{t\to\infty}\int_{\mathbb{R}^3}	\mathrm  dk f\left(t,k\right)   \omega(k) \chi_{B(O,\mathfrak R)}(k)
		\ = \ 0,\end{aligned}
\end{equation}
which is the cascade of energy towards infinity.

\end{proof}

\bibliographystyle{plain}

\bibliography{WaveTurbulence}

\end{document}